\documentclass[runningheads,a4paper]{llncs}
\usepackage{graphicx}
\usepackage[
	n,
	operators,
	advantage,
	sets,
	adversary,
	landau,
	probability,
	notions,	
	logic,
	ff,
	mm,
	primitives,
	events,
	complexity,
	asymptotics,
	keys]{cryptocode}

\usepackage{etoolbox}%
\newcommand{\FullVersion}[2]{%
    \ifstrequal{full}{full}
    {#1}
    {#2}
}
	
\usepackage{skak}
\usepackage{thm-restate}

\usepackage{dsfont}
\usepackage{latexsym}
\usepackage{graphicx}
\usepackage[breakable,skins]{tcolorbox}
\usepackage{enumitem}

\usepackage{todonotes}
\setuptodonotes{size=\scriptsize}

\usepackage[breaklinks=true]{hyperref} 
\usepackage{tikz}
\usetikzlibrary{shapes.callouts, decorations.pathreplacing, patterns, patterns.meta}
\usepackage{algorithm}
\usepackage[noend]{algpseudocode}
\usepackage[framemethod=tikz]{mdframed}
\usepackage{xspace}
\usepackage{pgfplots}
\usepackage{framed}
\pgfplotsset{compat=1.5}
\usepackage{standalone}

\usepackage{breakcites}

\usepackage{amsfonts}
\usepackage{mathtools}
\DeclarePairedDelimiter\pars{\lparen}{\rparen}

\spnewtheorem{numberedclaim}{Claim}{\bfseries}{\itshape}

\renewenvironment{claim}{\begin{numberedclaim}}{\end{numberedclaim}}

\DeclareMathAlphabet{\mathpzc}{OT1}{pzc}{m}{it}

\newenvironment{proofsketch}[1]{\begin{trivlist} \item {\it Proof
#1:~~}}
  {\qed\end{trivlist}}

\newcommand{\namedref}[2]{\hyperref[#2]{#1~\ref*{#2}}}
\newcommand{\eqnamedref}[2]{\hyperref[#2]{#1~(\ref*{#2})}}
\newcommand{\thmlab}[1]{\label{thm:#1}}
\newcommand{\thmref}[1]{\namedref{Theorem}{thm:#1}}
\newcommand{\lemlab}[1]{\label{lem:#1}}
\newcommand{\lemref}[1]{\namedref{Lemma}{lem:#1}}
\newcommand{\claimlab}[1]{\label{claim:#1}}
\newcommand{\claimref}[1]{\namedref{Claim}{claim:#1}}

\newcommand{\seclab}[1]{\label{sec:#1}}
\newcommand{\secref}[1]{\namedref{Section}{sec:#1}}
\newcommand{\applab}[1]{\label{app:#1}}
\newcommand{\appref}[1]{\namedref{Appendix}{app:#1}}

\newcommand{\figlab}[1]{\label{fig:#1}}
\newcommand{\figref}[1]{\namedref{Figure}{fig:#1}}
\newcommand{\alglab}[1]{\label{alg:#1}}
\renewcommand{\algref}[1]{\namedref{Algorithm}{alg:#1}}

\newcommand{\deflab}[1]{\label{def:#1}}
\newcommand{\defref}[1]{\namedref{Definition}{def:#1}}
\newcommand{\eqnlab}[1]{\label{eq:#1}}
\newcommand{\eqnref}[1]{\eqnamedref{Equation}{eq:#1}}

\newcommand{\propref}[1]{\namedref{Proposition}{prop:#1}}
\newcommand{\proplab}[1]{\label{prop:#1}}

\newenvironment{remindertheorem}[1]{\medskip {\bf Reminder of  #1.  }\em}{}
\newenvironment{reminderclaim}[1]{\medskip {\bf Reminder of  #1.  }\em}{}
\usepackage{skak}
\newcommand{\lcompress}{\mathtt{LCompress}}

\newcommand{\compress}{\mathtt{Compress}}
\newcommand{\kc}{\mathtt{KC}}
\newcommand{\ckc}{\mathtt{CKC}}
\newcommand{\score}{\mathsf{Score}}
\newcommand{\dpcompress}{\mathtt{DPCompress}}

\newcommand{\lapalg}{\mathcal{T}}
\newcommand{\decompress}{\mathtt{Decompress}}

\newcommand{\startinside}{\mathsf{StartInside}}

\newcommand{\B}{\mathcal{B}}
\newcommand{\M}{\mathcal{M}}
\newcommand{\BAD}{\mathsf{BAD}}

\newcommand{\GS}{\mathtt{GS}}   
\newcommand{\Lap}{\mathrm{Lap}}
\newcommand{\Ham}{\mathtt{Ham}}
\newcommand{\SR}{\mathtt{SR}}

\newcommand{\ctb}{\mathsf{ctb}}
\newcommand{\ctc}{\mathsf{ctc}}

\newcommand{\superscript}[1]{\ensuremath{^{\mbox{\scriptsize{\textit{#1}}}}}}
\def \th {\superscript{th}} 

\renewcommand{\S}{\mathcal{S}}

\newcommand{\Enc}{\mathsf{E}}
\newcommand{\QuinStr}{\mathsf{QuinStr}}
\newcommand{\str}{\mathsf{str}}

\def \A {\mdef{\mathcal{A}}}
\def \B {\mdef{\mathcal{B}}}

\renewcommand{\O}[1]{\ensuremath{\mathcal{O}\pars*{#1}}}

\newcommand{\mdef}[1]{{\ensuremath{#1}}\xspace}  

\def \th {\superscript{th}}     

\newcommand{\ignore}[1]{}

\newif\ifnotes\notestrue 
\ifnotes
\newcommand{\samson}[1]{\textcolor{purple}{{\bf (Samson:} {#1}{\bf ) }} \marginpar{\tiny\bf
             \begin{minipage}[t]{0.5in}
               \raggedright S:
            \end{minipage}}}  
\newcommand{\jeremiah}[1]{
\todo[color=crimson!30]{\textbf{Jeremiah:} #1}
}

\else
\newcommand{\samson}[1]{}
\newcommand{\jeremiah}[1]{}
\fi

\makeatletter
\renewcommand*{\@fnsymbol}[1]{\textcolor{blue}{\ensuremath{\ifcase#1\or *\or \dagger\or \ddagger\or
 \mathsection\or \triangledown\or \mathparagraph\or \|\or **\or \dagger\dagger
   \or \ddagger\ddagger \else\@ctrerr\fi}}}
\makeatother

\providecommand{\email}[1]{\href{mailto:#1}{\nolinkurl{#1}\xspace}}

\definecolor{mahogany}{rgb}{0.75, 0.25, 0.0}
\definecolor{darkblue}{rgb}{0.0, 0.0, 0.55}
\definecolor{darkpastelgreen}{rgb}{0.01, 0.75, 0.24}
\definecolor{darkgreen}{rgb}{0.0, 0.2, 0.13}
\definecolor{darkgoldenrod}{rgb}{0.72, 0.53, 0.04}
\definecolor{forestgreen}{rgb}{0.13, 0.55, 0.13}
\definecolor{darkred}{rgb}{0.55, 0.0, 0.0}
\definecolor{blueviolet}{RGB}{138,43,226}
\definecolor{dodgerblue}{RGB}{30,144,255}
\definecolor{crimson}{RGB}{220,20,60}
\hypersetup{
    colorlinks  = true,
    citecolor   = dodgerblue,
	linkcolor	= crimson
}	
\let\llncssubparagraph\subparagraph
\let\subparagraph\paragraph

\usepackage{titlesec}
\titlespacing*{\subsubsection}{0pt} {1.5ex}{1.5ex}

\let\subparagraph\llncssubparagraph
\usepackage{amsthm}

\begin{document}

\title{Differentially Private Compression and the Sensitivity of LZ77}

%
\titlerunning{Differentially Private Compression and the Sensitivity of LZ77}
%

\author{Jeremiah Blocki\inst{1}\orcidID{0000-0002-5542-4674}\index{Blocki, Jeremiah} 
\and Seunghoon Lee\thanks{Work done primarily while Seunghoon was at Purdue University.}\inst{2}\orcidID{0000-0003-4475-5686}\index{Lee, Seunghoon}
\and Brayan Sebastián Yepes-Garcia\inst{3}\orcidID{0009-0002-9162-390X}\index{Yepes-Garcia, Brayan Sebastián}
}
\authorrunning{J. Blocki, S. Lee, and S. Yepes-Garcia}
%
\institute{Purdue University, West Lafayette, USA\\ \email{jblocki@purdue.edu} \and University of Waterloo, Waterloo, Canada\\ \email{seunghoon.lee@uwaterloo.ca} \and Universidad Nacional de Colombia, Bogotá, Colombia\\ \email{byepesg@unal.edu.co}
}

%
\maketitle              

\begin{abstract}
We initiate the study of differentially private data-compression schemes motivated by the insecurity of the popular ``Compress-Then-Encrypt'' framework. Data compression is a useful tool which exploits redundancy in data to reduce storage/bandwidth when files are stored or transmitted. However, if the contents of a file are confidential then the \emph{length} of a compressed file might leak confidential information about the content of the file itself. Encrypting a compressed file does not eliminate this leakage as data encryption schemes are only designed to hide the \emph{content} of confidential message instead of the \emph{length} of the message. In our proposed \emph{Differentially Private Compress-Then-Encrypt} framework, we add a random positive amount of padding to the compressed file to ensure that any leakage satisfies the rigorous privacy guarantee of $(\epsilon,\delta)$-differential privacy. The amount of padding that needs to be added depends on the sensitivity of the compression scheme to small changes in the input, i.e., to what degree can changing a single character of the input message impact the length of the compressed file. While some popular compression schemes are highly sensitive to small changes in the input, we argue that effective data compression schemes do not necessarily have high sensitivity. Our primary technical contribution is analyzing the fine-grained sensitivity of the LZ77 compression scheme (IEEE Trans. Inf. Theory 1977) which is one of the most common compression schemes used in practice. We show that the global sensitivity of the LZ77 compression scheme has the upper bound $\mathcal{O}(W^{2/3}\log n)$ where $W\leq n$ denotes the size of the sliding window. When $W=n$, we show the lower bound $\Omega(n^{2/3}\log^{1/3}n)$ for the global sensitivity of the LZ77 compression scheme which is tight up to a sublogarithmic factor.

\end{abstract}

\section{Introduction}\seclab{intro}
Data compression algorithms exploit redundancy in data to compress files before storage or transmission. Lossless compression schemes are widely used in contexts such as ZIP archives (Deflate \cite{rfc1951}) and web content compression by Google (Brotli \cite{rfc7932}). In the Compress-Then-Encrypt framework, a message $w$ is first compressed to obtain a (typically shorter) file $y=\compress(w)$ and then encrypted to obtain a ciphertext $c=\Enc_K(y)$. Intuitively, compression is used to reduce bandwidth, and encryption is used to protect the file contents from eavesdropping attacks before it is transmitted over the internet.  

While encryption protects the {\em content} of the underlying plaintext message $y$, there is no guarantee that the ciphertext $c$ will hide the \emph{length} of the underlying plaintext $y$. For example, AES-GCM \cite{rfc5282} is the most widely used symmetric key encryption scheme, but the ciphertext leaks the {\em exact} length of the underlying plaintext\footnote{Any efficient encryption scheme must at least leak some information about the length of the underlying plaintext. Otherwise, the ciphertext of a short message (e.g., ``hi'') would need to be longer than the longest possible message supported by the protocol (e.g., a 4GB movie).}. Unfortunately, the length of the compressed message $y$ can leak information about the {\em content} of the original uncompressed file. For example, suppose that the original message was $w=$``\texttt{The top secret password is:~password.}'' An eavesdropping attacker who intercepts the AES-GCM ciphertext $c$ would be able to infer the length $\left|\compress(w)\right|$ and could use this information to eliminate many incorrect password guesses, i.e.,  if $\left| \compress(w')\right| \neq \left|\compress(w)\right|$ where $w'=$``\texttt{The top secret password is:~[guess]}'' then the attacker could immediately infer that this particular password guess is incorrect. 

The observation that the length of a compressed file can leak sensitive information about the {\em content} has led to several real-world exploits. For example, the CRIME (Compression Ratio Info-leak Made Easy) attack \cite{RizDuo12} exploits compression length leakage to hijack TLS sessions \cite{Fis12}. Similarly, the BREACH (Browser Reconnaissance and Exfiltration via Adaptive Compression of Hypertext) attack \cite{GluHarPra13} against the HTTPS protocol exploits compression in the underlying HTTP protocol. Despite its clear security issues, the  ``Compress-then-Encrypt'' framework continues to be used as a tool to reduce bandwidth overhead. There have been several heuristic proposals to address this information leakage. One proposal \cite{paulsen2019debreach} is static analysis to identify/exclude sensitive information before data compression. However, it is not always clear {\em a priori} what information should be considered sensitive. Heal the Breach \cite{palacios2022htb} pads compressed files by a random amount to mitigate leakage and protect against the CRIME/BREACH attacks. While this defense is intuitive, there will still be some information leakage and there are no rigorous privacy guarantees. This leads to the following question: 

\begin{quote}
    \emph{Can one design a compression scheme which provides rigorous privacy guarantees against an attacker who learns the compressed file’s length?}
\end{quote}

Differential Privacy (DP) \cite{TCC:DMNS06} has emerged as a gold standard in privacy-preserving data analysis due to its rigorous mathematical guarantees and strong composition results. Given $\epsilon>0$ and $\delta\in(0,1)$, a randomized algorithm $\A$ is called \emph{$(\epsilon,\delta)$-differentially private} if for every pair of neighboring datasets $\mathfrak{D}$ and $\mathfrak{D}'$ and for all sets $S$ of possible outputs, we have $\Pr[\A(\mathfrak{D})\in S]\leq e^\epsilon\cdot\Pr[\A(\mathfrak{D}')\in S]+\delta$. Typically, $\epsilon$ is a small constant and the additive loss $\delta$ should be negligible (if $\delta=0$ then we can simply say that $\A$ satisfies \emph{$\epsilon$-DP}). In our context, we are focused on leakage from the length of a compressed file. We say that an compression algorithm $\compress$ is  \emph{$(\epsilon,\delta)$-differentially private} if for every pair of neighboring strings $w$ and $w'$ (e.g., differing in one symbol) and all $S \subseteq \mathbb{N}$ we have  $\Pr[\left| \compress(w)\right|\in S]\leq e^\epsilon\cdot\Pr[\left| \compress(w')\right|\in S]+\delta$. Intuitively, suppose that $w$ denotes the original message and $w'$ denotes the message after replacing a sensitive character by a special blank symbol. Differential privacy ensures that an attacker who observes the length of the compressed file will still be unlikely to distinguish between $w$ and $w'$. DP composition ensures that the attacker will also be unlikely to distinguish between the original string $w$ and a string $w''$ where {\em every} symbol of a secret password has been replaced --- provided that the secret password is not too long.  

We first note that the Heal the Breach proposal \cite{palacios2022htb} does not necessarily satisfy differential privacy. For example, Lagarde and Perifel \cite{lagarde2018lempel} proved that the LZ78 compression algorithm \cite{LZ78} is highly sensitive to small bit changes. In particular, there exists nearby strings $w \sim w'$ (i.e., $|w|=|w'|$ and $|\{i:w[i]\neq w'[i]\}|=1$) of length $n$ such that $\left|\compress_{\sf LZ78}(w) \right| = o(n)$ and $\left|\compress_{\sf LZ78}(w') \right| = \Omega(n)$. In particular, to prevent an eavesdropping attacker from distinguishing between $w$ and $w'$, the length of the padding would need to be at least $\Omega(n)$. Heal the Breach \cite{palacios2022htb} adds a much smaller amount of padding $o(n)$.

\paragraph*{Notations.}
For a positive integer $n$, we define $[n]\coloneqq\{1,\ldots,n\}$. Unless otherwise noted, we will use $n$ as the length of a string throughout the paper. We use $\Sigma$ to denote a set of alphabets and $w \in \Sigma^n$ to denote a string of length $n$ and we let $w[i] \in \Sigma$ denote the $i$th character of the string. Given two strings $w,w' \in \Sigma^n$, we use $\Ham(w,w') = \left| \{ i:~w[i] \neq w'[i]\} \right|$ to denote the \emph{hamming distance} between $w$ and $w'$, i.e., the number of indices where the strings do not match. We say that two strings $w$ and $w'$ are \emph{adjacent} (or \emph{neighbors}) if $\Ham(w,w')=1$ and we use $w \sim w'$ to denote this. We use $|w|$ to denote the length of a string, i.e., for any string $w \in \Sigma^n$ we have $|w| = n$. We use $w \circ w'$ to denote the concatenation of two strings $w$ and $w'$. Given a character $c \in \Sigma$ and an integer $k \geq 1$, we use $c^k$ to denote the character $c$ repeated $k$ times. Formally, we define $c^1 \coloneqq c$ and $c^{i+1} \coloneqq c \circ c^i$. Unless otherwise noted, we assume that all log's have base $2$, i.e., $\log x \coloneqq \log_2 x$.

\subsection{Our Contributions}\seclab{contributions}
In this paper, we initiate the study of differentially private compression schemes focusing especially on the  LZ77 compression algorithm \cite{LZ77}. We first provide a general transformation showing how to make any compression scheme differentially private by adding a random amount of padding which depends on the global sensitivity (the worst-case difference in the compressed length when one input symbol is modified; see \defref{def:local_sensitivity} for the formal definition) of the compression scheme. The transformation will yield an efficient compression scheme as long as the underlying compression scheme has low global sensitivity. Second, we demonstrate that good compression schemes do not inherently have high global sensitivity by demonstrating that Kolmogorov compression has low global sensitivity. Third, we analyze the global sensitivity of the LZ77 compression scheme providing an upper/lower bound that is tight up to the sublogarithmic factor of $\log^{2/3} n$. 

\subsubsection{Differentially Private Transform for Compression Algorithms.}

 We provide a general framework to transform any compression scheme $\compress$ into a differentially private compression algorithm $\dpcompress(w,\epsilon,\delta)$ by adding a random positive amount of padding to $\compress(w)$, where the amount of padding $p$ depends on the privacy parameters $\epsilon$ and $\delta$ as well as the global sensitivity of the underlying compression algorithm $\compress$. More concretely, we show that $A_{\epsilon,\delta}(w)\coloneqq\dpcompress(w,\epsilon,\delta)$ is $(\epsilon,\delta)$-differentially private, i.e., the length leakage is $(\epsilon,\delta)$-differentially private.  The expected amount of padding added is $\O{\frac{\GS_\compress}{\epsilon}\ln(\frac{1}{2\delta})}$ (see \defref{def:local_sensitivity} for the definition of the global sensitivity $\GS_\compress$). Thus, as long as  $\GS_\compress \ll n$, it is possible to achieve a compression ratio $\left| \dpcompress(w,\epsilon,\delta) \right|/ n = \left| \compress(w) \right|/ n + o(1)$ that nearly matches $\compress$. See \secref{sec:dpcompress} for details.

\subsubsection{Idealized Compression Schemes have Low Sensitivity.}

While the LZ78 \cite{LZ78} compression algorithm has high global sensitivity, we observe that compression schemes achieving optimal compression ratios do not inherently have high global sensitivity. In particular, we argue that Kolmogorov compression has low global sensitivity, i.e., at most $\O{\log n}$ where $n$ denotes the string length parameter. While Kolmogorov compression is uncomputable, it is known to achieve a compression rate that is {\em at least as good} as any compression algorithm. We also construct a \emph{computable} variant of Kolmogorov compression that preserves the global sensitivity, i.e., the global sensitivity of the computable variant of Kolmogorov compression is also $\O{\log n}$.
See \secref{sec:kolmogorov} for details.

\subsubsection{Global Sensitivity of the LZ77 Compression Scheme.}

Our primary technical contribution is to analyze the global sensitivity of the  LZ77 compression algorithm \cite{LZ77}. The LZ77 algorithm includes a tunable parameter $W \leq n$ for the size of the sliding window. Selecting smaller $W$ reduces the algorithm's space footprint, but can result in worse compression rates because the algorithm can only exploit redundancy within this window. We provide an almost tight upper and lower bound of the global sensitivity of the LZ77 compression scheme. In particular, we prove that the global sensitivity of the LZ77 compression scheme for strings of length $n$ is upper bounded by $\mathcal{O}(W^{2/3}\log n)$ where $W \leq n$ is the length of the sliding window. When $W=n$ the global sensitivity is lower bounded by $\Omega(n^{2/3}\log^{1/3}n)$ matching our upper bound up to a sublogarithmic factor of $\log^{2/3}n$. We hope that our initial paper inspires follow-up research analyzing the global sensitivity of other compression schemes. 

\begin{theorem}[informal]
Let $\mathsf{LZ77}$ be the LZ77 compression scheme for strings of length $n$ with a sliding window size $W \leq n$. Then $\GS_\mathsf{LZ77} = \mathcal{O}(W^{2/3}\log n)$ and, when $W=n$, $\GS_\mathsf{LZ77} = \Omega(n^{2/3}\log^{1/3}n)$. (See \thmref{thm:thmupperbound} and \thmref{thm:lowerbound}.)
\end{theorem}

At a high level, the upper bound analysis considers the relationship between the blocks generated by running the LZ77  compression for two neighboring strings $w \sim w'$. Let $B_1,\ldots, B_t$ (resp. $B_1',\ldots, B_{t'}'$) denote the blocks generated when we run LZ77 with input $w$ (resp. $w'$) where block $B_i$ (resp. $B'_k$) encodes the substring $w[s_i, f_i]$ (resp. $w'[s_k',f_k'])$. We prove that if $s_i \leq s_k' \leq f_i$ (block $B_k'$ starts inside block $B_i$) then $f_{k+1}' \geq f_i$. In particular, this means that for every block $B_i$ there are \emph{at most two} blocks from $B_1',\ldots, B_{t'}'$ that ``start inside'' $B_i$. We can then argue that the difference in compression lengths is proportional to $t_2$ where $t_2 = \left| \left\{ i \leq t: \left|\{ k\leq t': s_i \leq s_k' \leq f_i  \}\right| =2\right\} \right|$ denotes the number of ``type-2'' blocks, i.e., $B_i$ which have two blocks that start inside it. Finally, we can upper bound $t_2 = \O{n^{2/3}}$ when $W=n$ or $t_2 = \O{W^{2/3}}$ when $W < n$. 

The lower bound works by constructing a string which (nearly) maximizes $t_2$. See \secref{sec:upperbound} and \secref{sec:lowerbound} for details.

\subsection{Related Work} 

There has been extensive work on developing efficient compression algorithms. Huffman coding \cite{Huffman52} encodes messages symbol by symbol --- symbols that are used most frequently are encoded by shorter binary strings in the prefix-free code. Lempel and Ziv \cite{LZ77,LZ78} developed multiple algorithms, and Welch \cite{Welch84} later proposed LZW as an improved implementation of \cite{LZ78}. Our work primarily focuses on analyzing the LZ77 compression algorithm \cite{LZ77} since it is one of the most common lossless compression algorithms used in practice. Deflate \cite{rfc1951} is a lossless compression scheme that combines LZ77 compression \cite{LZ77} and Huffman coding \cite{Huffman52} and is a key algorithm used in ZIP archives. The Lempel–Ziv–Markov chain algorithm (LZMA) is used in the 7z format of the 7-Zip archiver and it uses a modification of the LZ77 compression. Brotli \cite{rfc7932} also uses the LZ77 compression with Huffman coding and the 2nd-order context modeling. 

Several prior papers have studied the sensitivity of compression schemes \cite{lagarde2018lempel,AFI23,GILRSU23} to small input changes. Lagarde and Perifel \cite{lagarde2018lempel} studied the multiplicative sensitivity of the LZ78 compression algorithm \cite{LZ78}, and their result implies that the additive (i.e., global) sensitivity can be as large as $\Omega(n)$. Giuliani et al. \cite{GILRSU23} showed that the additive sensitivity of the Burrows-Wheeler Transform with Run-Length Encoding can be as large as $\Omega(\sqrt{n})$ --- upper bounding the additive sensitivity remains an open question. 

Most closely related to ours is the work of Akagi et al. \cite{AFI23} who studied the additive and multiplicative sensitivity of several compression schemes including Kolmogorov and LZ77. Akagi et al. \cite{AFI23} proved that Kolmogorov compression has additive sensitivity $\O{\log n}$. We extend this result to a computable variation of Kolmogorov compression in \secref{sec:kolmogorov}. For LZ77, they proved that the additive sensitivity is lower bounded by $\Omega(\sqrt{n})$ when the window size is $W=\Omega(n)$\footnote{Bannai et al. \cite{CPM:BanChaRad24} also lower bounds the additive sensitivity of LZ77 by $\Omega(\sqrt{n})$ as well though the primary focus in \cite{CPM:BanChaRad24} is on efficiently updating the compressed file after rotations instead of analyzing additive sensitivity.}. We prove that the additive sensitivity is lower bounded by $\widetilde\Omega(n^{2/3})$. Finally, they prove that the (local) additive sensitivity of LZ77 is {\em at most} $\O{z}$ where $z$ is the length of the compressed file. Unfortunately, this result does not even imply that the global sensitivity is $o(n)$ because $z$ can be as large as $z= \Omega(n)$ when the file is incompressible. 
By contrast, we prove that the global sensitivity is upper bounded by $\O{W^{2/3} \log n}$ which is tight up to logarithmic factors and is at most $\O{n^{2/3} \log n}$ even when $W = \Omega(n)$.

Degabriele \cite{CCS:Degabriele21} introduced a formal model for length-leakage security of compressed messages with random padding. While Degabriele \cite{CCS:Degabriele21} did not use differential privacy as a security notion, his analysis suggests that DP friendly distributions such as the  Laplace distribution and Gaussian distribution minimized the leakage. Song \cite{song2024refined} analyzed the (in)security of compression schemes against attacks such as cookie-recovery attacks and used the additive sensitivity of a compression scheme to upper bound the probability of successful attacks. Neither work \cite{CCS:Degabriele21,song2024refined} formalized the notion of a differentially private compression scheme as we do in \secref{sec:dpcompress}.  

Ratliff and Vadhan introduced a framework for differential privacy against timing attacks \cite{ratliff2024framework} to deal with information leakage that could occur when the running time of a (randomized) algorithm might depend on the sensitive input. They propose to introduce random positive delays after an algorithm is finished. The length of this delay will depend on the sensitivity of the algorithm's running time to small changes in the input. While our motivation is different, there are similarities: they analyze the sensitivity of an algorithm's running time while we analyze the sensitivity of a compression algorithm with respect to the length of the file. They introduce a random positive delay while proposing to add a random positive amount of padding to the compressed file. 

\section{Preliminaries}\seclab{sec:prelim}

\begin{definition}[Lossless Compression]
Let $\Sigma,\Sigma'$ be sets of alphabets. A {\em lossless compression scheme} consists of two functions $\compress:\Sigma^* \rightarrow (\Sigma')^*$ and $\mathtt{Decompress}:(\Sigma')^* \rightarrow \Sigma^*$. We require that for all string $w \in \Sigma^*$ we have $\mathtt{Decompress}\left(\compress(w) \right) = w$.
\end{definition}

\paragraph*{Global Sensitivity of Compression Scheme.}

Global sensitivity helps us understand how much the output of a function can change when its input is slightly modified. Recall that two strings $w,w' \in \Sigma^n$ are neighbors (denoted $w \sim w'$) if the Hamming Distance between these two strings is at most one, i.e., $\left|\{ i: w[i] \neq w'[i] \}  \right| \leq 1$. In the context of compression schemes, the global sensitivity of a compression scheme is defined as the largest difference in the compression length when we compress two neighboring strings $w \sim w'$.

\begin{definition}[Local/Global Sensitivity]
\deflab{def:local_sensitivity}
Let $\Sigma,\Sigma'$ be sets of alphabets. 
The \emph{local sensitivity} of the compression algorithm $\compress:\Sigma^* \rightarrow (\Sigma')^*$ at $w \in \Sigma^n$ is defined as $\mathtt{LS}_{\compress}(w) \coloneqq \max_{w' \in \Sigma^n:w \sim w'} \left| |\compress(w)| - |\compress(w')| \right|$, and the \emph{global sensitivity} of the compression algorithm for strings of length $n$ is defined as
$ \GS_{\compress}(n) \coloneqq \max_{w \in \Sigma^n} \mathtt{LS}_{\compress}(w)$. 
If it is clear from context, then one can omit the parameter $n$ (length of a string) and simply write $\GS_{\compress}$ to denote the global sensitivity.
\end{definition}

Understanding the global sensitivity of a compression algorithm is a crucial element of designing a differentially private compression scheme. While there is a large line of work using global sensitivity  \cite{sheng2025differentiallyprivatedistancequery,farias2025differentiallyprivatemultiobjectiveselection,wicker2024certificationdifferentiallyprivateprediction,10735286,10597947,zhang2023sensitivityestimationdifferentiallyprivate,tetek:LIPIcs.APPROX/RANDOM.2024.73,10.1145/3523227.3546781,9762326,blocki_et_al:LIPIcs.APPROX/RANDOM.2023.59,blocki_et_al:LIPIcs.ICALP.2022.26} to design differentially private mechanisms, to the best of our knowledge, no prior work has studied the construction of differentially private compression schemes.

In our context, the amount of random padding added to a compressed file will scale with the global sensitivity of the compression algorithm. While Lagarde and Perifel \cite{lagarde2018lempel} were not motivated by privacy or security, their analysis of LZ78 implies that the global sensitivity of this compression algorithm is $\Omega(n)$. Thus, to achieve differential privacy, the random padding would need to be very long, i.e., at least $\Omega(n)$. This would immediately negate any efficiency gains since, after padding, the compressed file would not be shorter than the original file!

\paragraph*{LZ77 Compression Scheme.} 

The LZ77 compression algorithm \cite{LZ77} takes as input a string $w \in \Sigma^n$ and outputs a sequence of blocks $B_1,\ldots, B_k$ where each block $B_i =[q_i, \ell_i, c_i]$ is a tuple consisting of two non-negative integers $q_i$ and $\ell_i$ and a character $c_i \in \Sigma$. Intuitively, if blocks $B_1,\ldots, B_{i-1}$ encode the first $\ctc$ characters of $w$ then the block $B_i =[q_i, \ell_i, c_i]$ encodes the next $\ell_i +1$ characters of $w$. In particular, if we have already recovered $w[1,\ctc]$ then the decoding algorithm can recover the substring $w[\ctc+1,\ctc+\ell_i+1] = w[q_i,q_i+\ell_i-1] \circ c_i$ can be recovered by copying $\ell_i$ characters from $w[1,\ctc]$ beginning at index $q_i$ and then appending the character $c_i$. The compression algorithm defines the first block as $B_1 = [0,0,w[1]]$ (where $w[i]$ denotes the $i\th$ character of $w$) and then initializes counters $\ctb=2$ and $\ctc=1$. Intuitively, the counter $\ctb$ indicates that we will output block $B_\ctb$ next and the parameter $\ctc$ counts the numbers of characters of $w$ that have already been encoded by blocks $B_1,\ldots, B_{\ctb-1}$. In general, to produce the $\ctb\th$ block we will find the longest prefix of $w[\ctc+1,n]$ that is a substring of $w[\max\{1,\ctc-W+1\},\ctc]$, where $W\leq n$ denotes the size of the sliding window --- the algorithm only stores $W$ most recent characters to save space. If the character $w[\ctc+1]$ is not contained in $w[\max\{1,\ctc-W+1\},\ctc]$ then we will simply set $B_\ctb=[0,0,w[\ctc+1]]$ and increment the counters $\ctb$ and $\ctc$. Otherwise, suppose that the longest such prefix has length $\ell_\ctb > 0$ then for some $z$ such that $\max\{0,\ctc-W\}+\ell_\ctb\leq z\leq \ctc$ we have $w[\ctc+1,\ctc+\ell_\ctb] = w[z-\ell_\ctb+1,z]$. Then we set $B_\ctb = [z-\ell_\ctb+1, \ell_\ctb, w[\ctc+\ell_\ctb+1]]$, increment $\ctb$, and update $\ctc = \ctc+\ell_\ctb+1$. We terminate the algorithm if $\ctc=n$ (i.e., the entire string has been encoded) and then output the blocks $(B_1,\ldots,B_{\ctb-1})$. See \figref{fig:LZ77} for a toy example of running the LZ77 compression for a string $w=``aababcdbabca"\in\Sigma^{12}$ with $\Sigma=\{a,b,c,d\}$ and $W=n$.

\begin{figure}[ht!]  
    \centering
    \begin{tikzpicture}
        \node[anchor=west] (w) at (0,0) {$w$};
        \node[draw,minimum height=0.5cm,label={north:$B_1$},anchor=west] (B1) at ($(w)+(1,0)$) {$a$};
        \node[draw,minimum height=0.5cm,label={north:$B_2$},anchor=west] (B2) at ($(B1.east)+(0.1,0)$) {$ab$};
        \node[draw,minimum height=0.5cm,label={north:$B_3$},anchor=west] (B3) at ($(B2.east)+(0.1,0)$) {$abc$};
        \node[draw,minimum height=0.5cm,label={north:$B_4$},anchor=west] (B4) at ($(B3.east)+(0.1,0)$) {$d$};
        \node[draw,minimum height=0.5cm,label={north:$B_5$},anchor=west] (B5) at ($(B4.east)+(0.1,0)$) {$babca$};
        \node[anchor=west] at ($(w.west)+(0,-0.8)$) {Initialize: $B_1=[0,0,a]$, $\ctb=2$, and $\ctc=1$};
        \node[anchor=west] at ($(w.west)+(0,-1.4)$) {Intermediate Steps:};
        \node[anchor=north west] (table) at ($(w.west)+(0,-1.5)$) {
        \begin{tabular}{ccccccc}
        $\ctb$ & $\ctc$ & $\ell_\ctb$ & $z$ & $B_\ctb$ & updated $\ctb$ & updated $\ctc$ \\
        \hline
        $2$ & $1$ & $1$ & $1$ & $[1,1,b]$ & $3$ & $3$ \\
        $3$ & $3$ & $2$ & $3$ & $[2,2,c]$ & $4$ & $6$ \\
        $4$ & $6$ & $0$ & N/A & $[0,0,d]$ & $5$ & $7$ \\
        $5$ & $7$ & $4$ & $6$ & $[3,4,a]$ & $6$ & $12$
        \end{tabular}
        };
        \node[anchor=west] (output) at ($(table.south west)+(0,-0.3)$) {Output: $(B_1=[0,0,a],B_2=[1,1,b],B_3=[2,2,c],B_4=[0,0,d],B_5=[3,4,a])$};
    \end{tikzpicture}
    \caption{Example of Running the LZ77 Compression for a string $w=``aababcdbabca$''.}
    \figlab{fig:LZ77}
\end{figure}

The decompression works straightforwardly as follows: (1) Given the compression $(B_1,\ldots,B_t)$, parse $B_1=[0,0,c_1]$ and initialize the string $w=``c_1"$. (2) For $i=2$ to $t$, parse $B_i=[q_i,\ell_i,c_i]$ and convert it into a string $v$: if $q_i=\ell_i=0$ then $v\coloneqq c_i$; otherwise, $v\coloneqq w[q_i,q_i+\ell_i-1]\circ c_i$. Then update $w\gets w\circ v$. 

The LZ77 compression scheme described above is called the \emph{non-overlapping (i.e., without self-referencing)} LZ77. In the non-overlapping LZ77, it is always the case that $q_i+ \ell_i-1\leq\ctc$ (so that the substring that we copy from the previous encoding does not overlap with the current encoding). We still obtain a valid compression scheme if we allow for the case $q_i+\ell_i -1 > \ctc$, and we refer to this scheme as LZ77 compression with self-referencing. We focus on the non-overlapping LZ77 in the main body of this paper. However, as we demonstrate in \appref{app:missingproofC}, our techniques extend to LZ77 with self-referencing.

\raggedbottom

\section{Differentially Private Compression}\seclab{sec:dpcompress}

This section presents a general framework that transforms any compression scheme $\compress:\Sigma^*\rightarrow(\Sigma')^*$ to another compression scheme called $\dpcompress:\Sigma^*\times\mathbb{R}\times\mathbb{R}\rightarrow(\Sigma')^*$ such that for any $\epsilon>0$ and $\delta \in(0,1)$, the algorithm  $A_{\epsilon,\delta}(w) \coloneqq \dpcompress(w,\epsilon,\delta)$ is a $(\epsilon,\delta)$-differentially private compression algorithm as defined formally in \defref{dpcompression}. 

\begin{definition} \deflab{dpcompression}
We say that a randomized compression algorithm $A:\Sigma^*\rightarrow(\Sigma')^*$  is  \emph{$(\epsilon,\delta)$-differentially private} if for every $n \geq 1$, every pair of neighboring strings $w \sim w' \in \Sigma^n$ and all $S \subseteq \mathbb{N}$ we have  \[ \Pr\left[\left| A(w)\right|\in S\right]\leq e^\epsilon\cdot\Pr\left[\left| A(w')\right|\in S\right]+\delta \ .\] 
\end{definition}

Intuitively, $\dpcompress(w,\epsilon,\delta)$ works by running $\compress(w)$ and adding a random amount of padding. In particular, $\dpcompress(w,\epsilon,\delta)$ outputs $\compress(w) \circ 0\circ 1^{p-1}$ where $p$ is a random variable defined as follows: $p=\max\left\{1, \lceil Z+k\rceil \right\}$: where $Z\sim\Lap(\GS_\compress/\epsilon)$ is a random variable following Laplace distribution with mean $0$ and scale parameter $\GS_\compress/\epsilon$ and $k=\frac{\GS_\compress}{\epsilon}\ln(\frac{1}{2\delta})+\GS_\compress +1$ is a constant (see \algref{algorithm:DPCompress}). The decompression algorithm works straightforwardly as it is easy to remove padding $0\circ 1^{p-1}$ even though one does not know $p$, i.e., given the compressed string $\dpcompress(w,\epsilon,\delta)=\compress(w)\circ0\circ1^{p-1}$, we could start removing $1$'s from the right until we see $0$. After removing the $0$ as well, we get $\compress(w)$. Now we could obtain $w$ by calling $\decompress(\compress(w))$.

Intuitively, the output $Z+\left|\compress(w) \right|$ preserves $\epsilon$-DP since this is just the standard Laplacian Mechanism and the output $\lceil Z+\left|\compress(w) \right| + k\rceil =  \left|\compress(w) \right| + \lceil Z+k\rceil$ also preserves $\epsilon$-DP since it can be viewed as post-processing applied to a DP output. Finally, note that by our choice of $k$ we have $\Pr[p \neq \lceil Z+k\rceil ] =  \Pr[Z+k \leq 0] \leq \delta$. It follows that the output $\left|\compress(w) \right| + p$ preserves $(\epsilon,\delta)$-DP, as shown in \thmref{theorem_dp}.

\begin{algorithm}[ht!]
\caption{$\dpcompress(w, \epsilon, \delta)$}
\begin{algorithmic}

    \State \textbf{Input:} $w,\epsilon,\delta$

    \State \textbf{Output:} Differentially Private Padded Data.

    \State $Z \gets \Lap\left(\frac{\GS_\compress}{\epsilon}\right)$

    \State $k$ $\gets$ $\frac{\GS_\compress}{\epsilon}\ln(\frac{1}{2\delta})+\GS_\compress +1$ 
    
    \State $p \gets \max\left\{1, \lceil Z+k\rceil \right\}$

    \State \textbf{Return} $\pad(\compress(w), p) \coloneqq \compress(w)\circ0\circ1^{p-1}$
\end{algorithmic}\alglab{algorithm:DPCompress}
\end{algorithm}

\newcommand{\dptheorem}{Define $A_{\epsilon,\delta}(w) \coloneqq \dpcompress(w,\epsilon,\delta)$ then, for any $\epsilon, \delta >0$, the algorithm $A_{\epsilon,\delta}$ is a  $(\epsilon,\delta)$-differentially private compression scheme. 
}

\begin{theorem}\thmlab{theorem_dp}
    \dptheorem
\end{theorem}

The proof of \thmref{theorem_dp} is straightforward using standard DP techniques and therefore is deferred to \FullVersion{\appref{app:dpcompress}}{the full version \cite{cryptoeprint:2025/1733}}.

On average, the amount of padding added is approximately $k=\frac{\GS_\compress}{\epsilon}\ln(\frac{1}{2\delta})+\GS_\compress +1$. If $k  = o(n)$ then it will still be possible for $\dpcompress$ to achieve efficient compression ratios, i.e., $\left| \dpcompress(w,\epsilon,\delta) \right|/ n = \left| \compress(w) \right|/ n + o(1)$. On the other hand, if $k = \Omega(n)$ then the padding alone will prevent us from achieving a compression ratio of $o(1)$. Thus, our hope is to find practical compression schemes with global sensitivity $\GS_\compress = o(n)$ --- LZ78~\cite{LZ78} does not satisfy these criteria \cite{lagarde2018lempel}. This motivates our study of the global sensitivity of the LZ77 compression scheme \cite{LZ77} --- see \secref{sec:upperbound} for details.

\section{Upper Bound for the Global Sensitivity of LZ77}\seclab{sec:upperbound}

Recall that in \secref{sec:dpcompress}, we provided the framework to convert any compression scheme to a $(\epsilon,\delta)$-DP compression scheme by adding a random amount of padding $p=\max\left\{1, \lceil Z+k\rceil \right\}$, where $Z\sim\Lap(\GS_\compress/\epsilon)$ and $k=\frac{\GS_\compress}{\epsilon}\ln(\frac{1}{2\delta})+\GS_\compress +1$ is a constant. We observed that as long as $k=o(n)$, we can still argue that $\dpcompress$ achieves efficient compression ratios, i.e., $\left| \dpcompress(w,\epsilon,\delta) \right|/ n = \left| \compress(w) \right|/ n + o(1)$. That was the motivation to find practical compression schemes with global sensitivity $o(n)$. We argue that the LZ77 compression scheme \cite{LZ77} satisfies this property. For simplicity of exposition, we will assume that $W=n$ in most of our analysis and then briefly explain how the analysis changes when $W < n$. In particular, we prove that the global sensitivity of the LZ77 compression scheme is $\O{W^{2/3}\log n}$ or $\O{n^{2/3}\log n}$ when $W=n$.

\subsection{Analyzing the Positions of Blocks}\seclab{position}

Recall that LZ77 \cite{LZ77} with the compression function $\compress:\Sigma^*\rightarrow(\Sigma')^*$ outputs a sequence of blocks $B_1,\ldots,B_t$ where each block is of the form $B_i=[q_i,\ell_i,c_i]$ such that $0\leq q_i,\ell_i< n$ are nonnegative integers and $c_i\in\Sigma$ is a character. This implies that for a string $w\in\Sigma^n$, it takes $2\lceil\log n\rceil + \lceil\log|\Sigma|\rceil$ bits to encode each block and this is the same for all the blocks. Therefore, the length of compression is proportional to the number of blocks $t$, i.e., we have $|\compress(w)| = t(2\lceil\log n\rceil + \lceil\log|\Sigma|\rceil)$. Let $B_1',\ldots, B_{t'}'$ denotes the blocks when compressing $w'$ instead of $w$. Our observation above tells us that to analyze the global sensitivity of the LZ77 compression scheme, it is crucial to understand the upper/lower bound of $t'-t$ (WLOG we can assume $t'\geq t$ since we can always change the role of $w$ and $w'$) where $t$ (resp. $t'$) is the number of blocks in $\compress(w)$ (resp. $\compress(w')$) for neighboring strings $w\sim w'\in\Sigma^n$.

To analyze the difference between the number of blocks $t'-t$, it is helpful to introduce some notation. First, since $w \sim w'$ we will use $j \leq n$ to denote the unique index such that $w[j] \neq w'[j]$ --- note that $w[i] = w'[i]$ for all $i \neq j$. Second, the block $B_i=[q_i,\ell_i,c_i]$ can be viewed intuitively as an instruction for the decompression algorithm to locate the substring $w[q_i,q_i+\ell_i-1]$ from the part of $w$ that we have already decompressed, copy this substring and append it to the end of the of the decompressed file followed by the character $c_i$. While the inputs $q_i$ and $\ell_i$ tell us where to copy {\em from} it is also useful to let $s_i \coloneqq 1+\sum_{r=1}^{i-1} (\ell_r +1)$ and $f_i \coloneqq \sum_{r=1}^{i} (\ell_r +1)$ denote the location where the block is copied {\em to}, i.e., we have $w[s_i,f_i]  = w[q_i,q_i+\ell_i-1] \circ c_i$. 

We say that block $B_k'$ starts inside block $B_i$ if $s_i\leq s_k'\leq f_i$ and we indicate this with the predicate  $\startinside(i,k)\coloneqq 1$. Otherwise, if $s_k' < s_i$ or $s_k' > f_i$ then we have  $\startinside(i,k)\coloneqq 0$. Our key technical insight is that if $\startinside(i,k)=1$ then for block $B_{k+1}'$ we must have $f_{k+1}' \geq f_k$. In particular, for later blocks $B'_{k'}$ with $k'>k$ we will have $s'_{k'} > f_i$ so $\startinside(i,k')=0$. In particular, if we let $\M_i \coloneqq \{ k \in [t']~:~\startinside(i,k)=1\}$ then   \lemref{lemma:start_inside} tells us that one of three cases applies: (1) $\M_i = \emptyset$, (2) $\M_i = \{k\}$ for some $k \leq t'$, or (3) $\M_i = \{k,k+1\}$ for some $k < t'$. In any case, we have $\left| \M_i \right| \leq 2$. This has already been proven by Akagi et al. \cite{AFI23}, but for the sake of completeness, we include our proof in \FullVersion{\appref{app:missingproofA}}{the full version \cite{cryptoeprint:2025/1733}}.

\newcommand{\lemstartinsidestatement}{
Let $\compress:\Sigma^*\rightarrow(\Sigma')^*$ be the LZ77 compression algorithm and $w,w'\in\Sigma^n$ such that $w\sim w'$. Let $(B_1,...,B_t)\gets\compress(w)$ and $(B'_1,...,B'_{t'})\gets\compress(w')$. Then for all $i\in[t]$, either $\M_i = \emptyset$ or $\M_i=[i_1,i_2]$ for some $i_1\leq i_2\leq i_1+1$. In particular, $|\mathcal{M}_i|\leq 2, \forall i \in [t]$.
}
\begin{lemma}[{\cite{AFI23}}]\lemlab{lemma:start_inside}
    \lemstartinsidestatement
\end{lemma}



Now we can partition the blocks $B_1,\ldots,B_t$ into three sets based on the size of $\M_i$. In particular, let $\B_m \coloneqq \{ i ~:~|\M_i| = m\}$ for $m \in \{0,1,2\}$ --- \lemref{lemma:start_inside} implies that $\B_m = \emptyset$ for $m \geq 3$. To upper bound $t'-t$, it is essential to \emph{count the number of type-2 blocks}. Let $t_m = \left| \B_m  \right|$ be the number of type-$m$ blocks for $m=0,1,2$. Then we have the following claim.

\begin{claim}\claimlab{claim:set:blocksm2:GS}
Let $\compress:\Sigma^*\rightarrow(\Sigma')^*$ be the LZ77 compression function and $w,w'$ be strings of length $n$ and $w\sim w'$. Let $(B_1,\ldots,B_t)\gets\compress(w)$ and $(B_1',\ldots,B_{t'}')\gets\compress(w')$. Then $t'-t \leq t_2$.
\end{claim}
\begin{proof}
We observe $t_0+t_1+t_2=t$ by \lemref{lemma:start_inside} and since there are $m$ blocks in $(B'_1,\ldots,B'_{t'})$ that start inside type-$m$ blocks in $(B_1,\ldots,B_t)$ for $m=0,1,2$, we have $t'=\sum_{m=0}^2 m\cdot t_m = t_1 + 2t_2$, which implies that $t'-t = t_2 - t_0 \leq t_2$.
\end{proof}

\claimref{claim:set:blocksm2:GS} implies that $\left| \left| \compress(w) \right| - \left| \compress(w') \right| \right| \leq (t'-t)(2\left\lceil \log n \right \rceil + \left\lceil \log |\Sigma|\right \rceil)\leq  t_2(2\left\lceil \log n \right \rceil + \left\lceil \log |\Sigma|\right \rceil)$ since it takes $2\left\lceil \log n \right \rceil + \left\lceil \log |\Sigma|\right \rceil$ bits to encode each block (see \claimref{claim:compress:GS} in \appref{app:missingproofA}). Hence, upper bounding the number of type-2 blocks $t_2$ would allow us to upper bound the global sensitivity of the LZ77 compression.

\subsection{Counting Type-2 Blocks using Uniqueness of Offsets}\seclab{typetwoblocks}

We can effectively count the number of type-2 blocks by considering the location of the unique index $j$ such that $w[j]\neq w'[j]$ and show that the number of type-2 blocks is at most $\mathcal{O}(n^{2/3})$ as stated in \lemref{lemma:length_cases_total}. 

\newcommand{\lemblocknumstatement}{
Let $\compress:\Sigma^*\rightarrow(\Sigma')^*$ be the LZ77 compression function and $w,w'$ be strings of length $n$ and $w\sim w'$. Let $(B_1,\ldots,B_t)\gets\compress(w)$ and $(B_1',\ldots,B_{t'}')\gets\compress(w')$. Then $t_2\leq \frac{\sqrt[3]{9}}{2} n^{2/3} + \frac{\sqrt[3]{3}}{2} n^{1/3} + 1$.
}
\begin{lemma}\lemlab{lemma:length_cases_total}
    \lemblocknumstatement
\end{lemma}

\begin{proof}[Proof Sketch]
We only give the proof sketch here and the formal proof can be found in \appref{app:missingproofC}. To prove \lemref{lemma:length_cases_total}, we show the following helper claims. 
\begin{itemize}
    \item First, we show that if $B_i$ is a type-2 block then either $s_i\leq j\leq f_i$ or $q_i\leq j<q_i+\ell_i$ should hold (see \claimref{claim:length_cases_total_a} in \appref{app:missingproofC}). Intuitively, any type-2 block $B_i$ with $s_i > j$ is constructed by copying the longest possible substring \emph{from} a section around the index $j$ i.e., $j \in [q_i, q_i+\ell_i)$. Type-2 blocks cannot occur before the strings diverge at index $j$ i.e., if $f_i < j$ then $B_i$ is not a type-2 block.
    \item Second, we show that if blocks $B_{i_1} = [q_{i_1}, \ell_{i_1}, c_{i_1}]$ and $B_{i_2}= [q_{i_2}, \ell_{i_2}, c_{i_2}]$ are both type-2 blocks with $s_{i_1}, s_{i_2} > j$ then $(q_{i_1},\ell_{i_1})\neq(q_{i_2},\ell_{i_2})$ (see \claimref{claim:length_cases_total_b} in \appref{app:missingproofC}).  In particular, the pair $(q_i,\ell_i)$ must be {\em unique} for each type-2 block and, if $s_i > j$, then we must have $j \in [q_i, q_i+\ell_i)$ by our observation above.
    
    \item Finally, followed by the uniqueness of offsets from the previous claim, we can show that the number of type-2 blocks with length $\ell$ is at most $\ell$ by pigeonhole principle, i.e., if $\B_2^\ell\coloneqq\{i\in\B_2:\ell_i=\ell\}$ and if $s_{i^*}\leq j\leq f_{i^*}$ for some $i^*\in[t]$, then $|\B_2^\ell|\leq\ell$ for all $\ell\neq\ell_{i^*}$ and $|\B_2^{\ell_{i^*}}|\leq\ell_{i^*}+1$ (see \claimref{claim:length_cases_total_c} in \appref{app:missingproofC}).
\end{itemize}
Let $x_\ell \coloneqq |\B_2^\ell|$ denote the number of type-2 blocks with length $\ell$. The total number of type-2 blocks is given by the sum $\sum_\ell x_{\ell}$. If we try to maximize $\sum_\ell x_{\ell}$ subject to the constraints that $x_\ell \leq \ell$ (for all $\ell\neq\ell_{i^*}$), $x_{\ell_{i^*}}\leq \ell_{i^*}+1$, and $\sum_\ell x_\ell  (\ell+1) \leq n$,  we obtain a solution where we set $x_\ell = \ell$ for $\ell \leq z$ with $\ell\neq \ell_{i^*}$ (and set $x_{\ell_{i^*}}=\ell_{i^*}+1$ if $\ell_{i^*}\leq z$) and $x_\ell = 0$ for $\ell > z$ by a simple swapping argument (see the formal proof in \appref{app:missingproofC} for details). Note that $x_z$ could be less than $z$, i.e., for the threshold $z$, we have $x_z=c$ for some $0<c\leq z$.
To find the threshold $z$, we observe that $\left[\sum_{\ell=0}^{z-1} \ell(\ell+1)\right] + x_z(z+1) = \sum_{\ell \geq 0} x_\ell (\ell+1) \leq n$. Simplifying the sum, we have $\left[\sum_{\ell=0}^{z-1} \ell(\ell+1)\right] + x_z(z+1)=\frac{1}{3}(z-1)z(z+1) + x_z(z+1) \geq z^3/3$ where the last inequality follows because  $x_z\geq 1$. Thus, we have $z^3\leq 3n$ and $z\leq\sqrt[3]{3n}$.
Setting this value of $z$, we have $t_2 \leq \left(\sum_{\ell=0}^{\sqrt[3]{3n}} \ell\right)+1 =  \frac{\sqrt[3]{9}}{2} n^{2/3} + \frac{\sqrt[3]{3}}{2} n^{1/3} + 1$.
\end{proof}

While we focus on analyzing the global sensitivity of the \emph{non-overlapping} LZ77 compression scheme, one can easily extend the results to the LZ77 with self-referencing and achieve the same upper bound for $t_2=\O{n^{2/3}}$. See \appref{app:missingproofCSR} for details.

We remark that this is a key result to upper bound the global sensitivity of the LZ77 compression scheme from our previous discussion that it takes $\O{\log n}$ bits to encode each block. Since the number of type-2 blocks is $\mathcal{O}(n^{2/3})$ and $t'-t$ is bounded by the number of type-2 blocks, we can combine those results and conclude that the global sensitivity of the LZ77 compression scheme is upper bounded by $\mathcal{O}(n^{2/3}\log n)$, as stated in \thmref{thm:thmupperbound}.

\begin{theorem}\thmlab{thm:thmupperbound}
    Let $\compress:(\Sigma)^*\rightarrow(\Sigma')^*$ be the LZ77 compression function with unbounded sliding window size $W=n$. Then 
    \begin{equation*}
    \GS_\compress \leq \left(\frac{\sqrt[3]{9}}{2} n^{2/3} + \frac{\sqrt[3]{3}}{2} n^{1/3} + 1\right) (2\lceil\log n\rceil+\lceil\log|\Sigma|\rceil) = \O{n^{2/3}\log n}.    
    \end{equation*}
\end{theorem}
\begin{proof}
    Let $(B_1,\ldots,B_t)\gets\compress(w)$ and $(B'_1,\ldots,B'_{t'})\gets\compress(w')$.
    By \claimref{claim:set:blocksm2:GS} and \lemref{lemma:length_cases_total}, we have $t'-t\leq t_2\leq \frac{\sqrt[3]{9}}{2} n^{2/3} + \frac{\sqrt[3]{3}}{2} n^{1/3} + 1$. We know $|\compress(w)|=t(2\lceil\log n\rceil+\lceil\log|\Sigma|\rceil)$ and $|\compress(w')|=t'(2\lceil\log n\rceil+\lceil\log|\Sigma|\rceil)$. Hence,
    $\left|\left| \compress(w) \right| - \left| \compress(w') \right|\right| = |t-t'|(2\lceil\log n\rceil+\lceil\log|\Sigma|\rceil)
    \leq \left(\frac{\sqrt[3]{9}}{2} n^{2/3} + \frac{\sqrt[3]{3}}{2} n^{1/3} + 1\right) (2\lceil\log n\rceil+\lceil\log|\Sigma|\rceil)$.
Since this inequality hold for arbitrary $w\sim w'$ of length $n$, we have
\begin{align*}
    &\mathtt{GS}_{\mathtt{Compress}} = \max_{w \in \Sigma^n}\max_{w' \in \Sigma^n:w\sim w'} \left| |\compress(w)| - |\compress(w')| \right|\\
    &\quad\leq \left(\frac{\sqrt[3]{9}}{2} n^{2/3} + \frac{\sqrt[3]{3}}{2} n^{1/3} + 1\right) (2\lceil\log n\rceil+\lceil\log|\Sigma|\rceil)=\O{n^{2/3}\log n}.\qedhere
\end{align*}
\end{proof}

For the LZ77 with self-referencing, \lemref{lemma:start_inside} essentially still holds except that we might have one type-3 block (i.e., $t_3 = 1$), i.e., the block containing the edited position $j$ --- see \cite[p.18]{AFI23}. Thus, the global sensitivity of the LZ77 with self-referencing would be bounded by $\O{t_2 \log n + 2 t_3 \log n} = \O{t_2 \log n + 2\log n}=\O{n^{2/3}\log n}$. See \thmref{thm:thmupperboundsr} in \appref{app:missingproofCSR} for details.

\subsection{Bounded Sliding Window ($W < n$)}\seclab{boundedwindow}
The analysis for the case when $W<n$ is very similar. When $W < n$ then we obtain an additional constraint on any type-2 block as stated in \claimref{claim:newconstraint}. 

\newcommand{\newconstraint}{
Let $\compress:(\Sigma)^*\rightarrow(\Sigma')^*$ be the LZ77 compression function with sliding window size $W$ and $w\sim w'$ be strings of length $n$ with $w[j]\neq w'[j]$. Let $(B_1,\ldots,B_t)\gets\compress(w)$ and $(B'_1,\ldots,B'_{t'})\gets\compress(w')$. If $B_i\in\B_2$ then $s_i\leq j+W$.
}
\begin{claim}\claimlab{claim:newconstraint}
    \newconstraint
\end{claim}

The proof of \claimref{claim:newconstraint} is elementary and can be found in \appref{app:missingproofW}. \claimref{claim:newconstraint} tells us that there is no type-2 block that starts after $j+W$. Prior constraints still apply: we still have $f_i \geq j$ for type-2 blocks $B_i$ and if $s_i \geq j$ then we must have $j \in [q_i,q_i+\ell_i)$. In particular, we can have at most $\ell$ type-2 blocks of length $\ell$ and, if $s_i \geq j$, we must have $j \in [q_i,q_i+\ell_i)$. Letting $x_\ell$ denote the number of type-2 blocks of length $\ell$ (excluding any   block $B_i$ such that $j \in [s_i,f_i]$ or such that $f_i > j+W$ if any such type-2 blocks exist\footnote{Observe that we exclude {\em at most} two blocks since there can be {\em at most} one type-2 block $B_i$ with $f_i > j+W$  and there can be {\em at most} one type two block $B_i$ with $j \in [s_i,f_i]$. }) we now obtain the constraint that $\sum_{\ell} (\ell+1) x_\ell \leq 3W$ instead of the prior constraint that $\sum_{\ell} (\ell+1) x_\ell \leq n$. Maximizing $\sum_\ell x_\ell$ subject to the above constraint, as well as the prior constraints that $x_\ell \leq \ell $ for all $\ell$, we can now obtain a tighter upper bound $t_2 = \O{W^{2/3}}$ instead of $t_2 = \O{n^{2/3}}$. This leads to the following result.

\begin{theorem}\thmlab{thm:thmupperboundW}
    Let $\compress:(\Sigma)^*\rightarrow(\Sigma')^*$ be the LZ77 compression function with sliding window size $W$. Then $\GS_\compress\leq \left(\frac{\sqrt[3]{81}}{2} W^{2/3} + \frac{\sqrt[3]{9}}{2} W^{1/3} + 3\right) (2\lceil\log n\rceil+\lceil\log|\Sigma|\rceil) = \O{W^{2/3}\log n}$.
\end{theorem}

Essentially the same result holds for the LZ77 with self-referencing as well; see \thmref{thm:thmupperboundsrW} in \appref{app:missingproofW} for details.

In practice, many implementations of LZ77 set $W \ll n$ to minimize space usage. In addition to minimizing space usage, these implementations also reduce the global sensitivity of the compression algorithm. 

\section{Lower Bound for the Global Sensitivity of LZ77}\seclab{sec:lowerbound}

In \secref{sec:upperbound}, we proved that the upper bound for the global sensitivity of the LZ77 compression algorithm $\compress$ is $O(n^{2/3}\log n)$ with window size $W=n$. One could ask if this is a tight bound, i.e., if we can prove the matching \emph{lower bound} for the global sensitivity of $\compress$ as well. This section proves the almost-matching lower bound up to a sub-logarithmic factor. In particular, we show that the global sensitivity of the LZ77 compression algorithm is $\Omega(n^{2/3}\log^{1/3} n)$. To prove the lower bound, we need to give example strings $w\sim w'$ of length $n$ that achieves $\left||\compress(w)|-|\compress(w')|\right|=\Omega(n^{2/3}\log^{1/3} n)$, since this implies $\displaystyle\GS_\compress =\max_{x \in \Sigma^n}\max_{x' \in \Sigma^n:x\sim x'} ||\compress(x)| \allowbreak- |\compress(x')||\geq \left||\compress(w)|-|\compress(w')|\right|=\Omega(n^{2/3}\log^{1/3} n)$. 
In this section, we will give the construction of such example strings $w$ and $w'$.

\subsection{String Construction}

Consider an encoding function $\Enc:\mathbb{Z}\rightarrow\bin^*$ that maps integers to binary strings. Then for a positive integer $m\in\mathbb{Z}$, we have an injective encoding of the number set $\mathcal{S}\coloneqq\left\{0,1,\ldots, m\right\}$ using $\lceil\log m\rceil$ bits, i.e., $\Enc(i)\neq\Enc(j)$ if $i,j\in\mathcal{S}$ and $i\neq j$. For example, if $m=2^q-1$ for some positive integer $q$, we could encode the elements of $\S$ as follows: 
\begin{equation*}
\Enc(0)=0^{\lceil\log m\rceil}, \Enc(1)=0^{\lceil\log m\rceil-1}1,\Enc(2)=0^{\lceil\log m\rceil-2}10,\ldots,\Enc(m)=1^{\lceil\log m\rceil}.
\end{equation*} 
Now, consider a quinary alphabet $\Sigma=\{0,1,2,3,4\}$ and define a string
\begin{equation*}
S_{\ell,u}\coloneqq\Enc(m-u+1)^2\circ\Enc(m-u+2)^2\circ\cdots\circ\Enc(m)^2\circ 2 \circ\Enc(m+1)^2\circ\cdots\circ\Enc(m-u+\ell)^2
\end{equation*}
in $\Sigma^*$ for $2\leq\ell\leq m$ and $1\leq u\leq\ell-1$. 
Here, $(\cdot)^2$ denotes the concatenation of the string itself twice, i.e., $\Enc(\cdot)^2=\Enc(\cdot)\circ\Enc(\cdot)$. 
We define a procedure called $\QuinStr(m)$ which takes as input a positive integer $m\in\mathbb{Z}$ and outputs two quinary strings as follows.

\begin{tcolorbox}[breakable,enhanced,title={The Construction of Two Quinary Strings $\QuinStr(m)$.}]
    \begin{enumerate}[leftmargin=*,label=\arabic*.]
        \item Compute two quinary strings $S_w$ and $S_{w'}$ where
        \begin{align*}
            S_w &\coloneqq\Enc(1)^2\circ\cdots\circ\Enc(m)^2\circ 2 \circ\Enc(m+1)^2\circ\cdots\circ\Enc(2m)^2\circ 4,\text{ and}\\
            S_{w'} &\coloneqq\Enc(1)^2\circ\cdots\circ\Enc(m)^2\circ 3 \circ\Enc(m+1)^2\circ\cdots\circ\Enc(2m)^2\circ 4.
        \end{align*}
        \item Compute two quinary strings $w\coloneqq S_w\circ S$ and $w'\coloneqq S_{w'}\circ S$, where
        \[S = S_{2,1} \circ 4 \circ S_{3,2} \circ 4 \circ S_{3,1} \circ 4 \circ \ldots \circ S_{m,m-1} \circ 4 \circ \ldots \circ S_{m,1}\circ 4.\]
        \item Output $(w,w')$.
    \end{enumerate}
\end{tcolorbox}

\claimref{claim:length} tells us that the strings $w$ and $w'$ outputted by the procedure $\QuinStr(m)$ has equal length $\Theta(m^3\log m)$. Since the proof is elementary, we defer the proof of \claimref{claim:length} to \appref{app:missingprooflowerbound}.

\newcommand{\claimlength}{
Let $m\in\mathbb{N}$ and $(w,w')\gets\QuinStr(m)$. Then $|w|=|w'|=\Theta(m^3\log m)$. In particular, for $m\geq 4$, $\frac{2}{3}m^3\lceil\log m\rceil< |w|=|w'|< m^3\lceil\log m\rceil$.
}
\begin{claim}\claimlab{claim:length}
\claimlength
\end{claim}

\subsection{Analyzing the Sensitivity of $\QuinStr(m)$}

A central step in our sensitivity analysis for $\QuinStr(m)$ is precisely counting the type-2 blocks produced by the LZ77 compression scheme, as we did in \secref{sec:upperbound}. \lemref{lem:gs} shows that for $(w,w')\gets\QuinStr(m)$, we have $|\B_2|=\frac{(m-1)m}{2}-(\lfloor\frac{m}{2}\rfloor-1)$. Intuitively, we first show that for $w=S_w\circ S$, there is no type-2 block for the blocks compressing $S_w$. Then the main insight is that we carefully crafted strings $w$ and $w'$ such that the marker symbol `4' becomes the endpoint for each block in $\compress(w)$ for the tail part $S$ of $w=S_w\circ S$. By repeating each encoding twice, we can ensure that most of the occurrences of $S_{\ell,u}\circ 4$ yield type-2 blocks, with an edge case (addressed in \claimref{claim:repeat}) that makes the block in $\B_1$ but this happens for only about $m/2$ blocks. Consequently, despite these few exceptions, the overall count of type-2 blocks remains quadratic in $m$.

\newcommand{\lemGS}{
Let $m\in\mathbb{N}$ and $(w,w')\gets\QuinStr(m)$ and let $\compress:\Sigma^*\rightarrow(\Sigma')^*$ be the LZ77 compression algorithm. Let $(B_1,\ldots,B_t)\gets\compress(w)$ and $(B'_1,\ldots,B'_{t'})\gets\compress(w')$. Then $|\B_0| = 0$ and $|\B_2| = \frac{(m-1)m}{2}-(\lfloor\frac{m}{2}\rfloor-1)$.
}
\begin{lemma}\lemlab{lem:gs}
    \lemGS
\end{lemma}

\begin{proof}
Recall that $w = S_w \circ S$ and $w' = S_{w'} \circ S$, where
\begin{itemize}
    \item $S_w =\Enc(1)^2\circ\cdots\circ\Enc(m)^2\circ 2 \circ\Enc(m+1)^2\circ\cdots\circ\Enc(2m)^2\circ 4$,
    \item $S_{w'}=\Enc(1)^2\circ\cdots\circ\Enc(m)^2\circ 3 \circ\Enc(m+1)^2\circ\cdots\circ\Enc(2m)^2\circ 4$, and
    \item $S=S_{2,1} \circ 4 \circ S_{3,2} \circ 4 \circ S_{3,1} \circ 4 \circ \ldots \circ S_{m,m-1} \circ 4 \circ \ldots \circ S_{m,1}\circ 4$, where
    \item $S_{\ell,u}\coloneqq\Enc(m-u+1)^2\circ\Enc(m-u+2)^2\circ\cdots\circ\Enc(m)^2\circ 2 \circ\Enc(m+1)^2\circ\cdots\circ\Enc(m-u+\ell)^2$ for $2\leq\ell\leq m$ and $1\leq u\leq \ell-1$.
\end{itemize}
We first observe that $w\sim w'$. Define $S_w^F\coloneqq \Enc(1)^2\circ\cdots\circ\Enc(m)^2\circ 2$ (resp. $S_{w'}^F\coloneqq \Enc(1)^2\circ\cdots\circ\Enc(m)^2\circ 3$) to be the first-half substring of $S_w$ (resp. of $S_{w'}$), and $S_w^L=S_{w'}^L\coloneqq \Enc(m+1)^2\circ\cdots\circ\Enc(2m)^2\circ 4$ to be the last-half substring of $S_w$ (or $S_{w'}$ since they are indeed identical). 
It is useful to define a notation $\str(B_k)$ for a block $B_k$, which denotes the substring of $w$ represented by the block $B_k$, i.e., for $B_k=[q_k,\ell_k,c_k]$, $\str(B_k)\coloneqq w[q_k,q_k+\ell_k-1]\circ c_k$. 

Let $B_{i_1}$ be the first block such that $S_w^F$ becomes a substring of $\str(B_1)\circ\str(B_2)\circ\ldots\circ\str(B_{i_1})$, and similarly, let $B'_{i'_1}$ be the first block such that $S_{w'}^F$ becomes a substring of $\str(B'_1)\circ\str(B'_2)\circ\ldots\circ\str(B'_{i'_1})$. Then we observe the following:
\begin{enumerate}
    \item $B_{i_1}=[q_{i_1},\ell_{i_1},2]$, i.e., $\str(B_{i_1})$ ends with $2$ (which is the last character in $S_w^F$), since $2$ never showed up before as all the encodings are binary strings, it has to be added to the dictionary as a new character,
    \item $B_i=B_i'$ for all $i\in[i_1-1]$, as we are compressing the identical strings until we see $2$ in $S_w^F$ (and $3$ in $S_{w'}^F$), and
    \item $i_1=i_1'$ and $B'_{i_1}=[q_{i_1},\ell_{i_1},3]$, since two strings $S_w^F$ and $S_{w'}^F$ are identical except for the very last character.\label{item:2}
\end{enumerate}

Now let $B_{i_2}$ be the first block such that $S_w^L$ becomes a substring of $\str(B_{i_1+1})\circ\str(B_{i_1+2})\circ\ldots\circ\str(B_{i_2})$, and similarly, let $B'_{i'_2}$ be the first block such that $S_{w'}^L$ becomes a substring of $\str(B'_{i_1+1})\circ\str(B'_{i_1+2})\circ\ldots\circ\str(B'_{i'_2})$. Then we observe the following:
\begin{enumerate}
\setcounter{enumi}{3}
    \item $i_2=i_2'$ and $B_i=B_i'$ for all $i\in[i_1+1,i_2]$, since $i_1=i_1'$ from observation \ref{item:2} and we have $S_w^L=S_{w'}^L$ while they do not contain $2$ or $3$, and
    \item $B_{i_2}=[q_{i_2},\ell_{i_2},4]$, since $4$ never showed up before in our compression.
\end{enumerate}

From the observations above, we have that $B_i\in\B_1$ for all $i\in[i_2]$. 
Now we are left with the blocks $(B_{i_2+1},\ldots,B_t)$ compressing the last part $S$ of $w$ and the blocks $(B'_{i_2+1},\ldots,B'_{t'})$ compressing the last part $S$ of $w'$. For the blocks $(B_{i_2+1},\ldots,B_t)$, we observe that each block ends at the next `$4$' because each $S_{\ell,u}$ (for $2\leq\ell\leq m$ and $1\leq u\leq \ell-1$) is contained in the former part of $w$ (which was $S_w$) but $4$ only shows up in $S_w$ followed by $\Enc(2m)^2$ while $S_{\ell,u}$ cannot contain $\Enc(2m)$. Hence, we observe the following:
\begin{enumerate}
\setcounter{enumi}{5}
    \item $\str(B_{i_2+1})=S_{2,1}\circ 4,\str(B_{i_2+2})=S_{3,2}\circ 4$,  and so on. To see this, we recall that $S_{\ell,u}\coloneqq\Enc(m-u+1)^2\circ\Enc(m-u+2)^2\circ\cdots\circ\Enc(m)^2\circ 2 \circ\Enc(m+1)^2\circ\cdots\circ\Enc(m-u+\ell)^2$. Because of `2' in the middle of each $S_{\ell,u}$ and `4' at the end of each $S_{\ell,u}\circ4$, we observe that $S_{\ell,u}\circ4$ cannot be copied from a previous position in $S$.
    \item In general, $\str\left(B_{i_2+\frac{(\ell-2)(\ell-1)}{2}+(\ell-u)}\right)=S_{\ell,u}\circ 4$, for $2\leq\ell\leq m$ and $1\leq u\leq \ell-1$. Intuitively, the blocks are numbered by iterating over $\ell=2,\ldots,m$, and for each $\ell$, iterating over $u=\ell-1,\ldots,1$ in reverse order. 
    This indeed covers all the blocks from $B_{i_2+1},\ldots,B_t$ (See \claimref{claim:inj} and observation \ref{item:8}).
    \item Furthermore, we can observe that $t= i_2 + (1+2+\ldots+(m-1)) = i_2 + \frac{(m-1)m}{2}$.\label{item:8}
\end{enumerate}

\newcommand{\felluinjective}{
For any integer $m\geq 2$, the function $f(\ell,u)\coloneqq\frac{(\ell-2)(\ell-1)}{2}+(\ell-u)$ defined over integers $\ell$ and $u$ such that $2\leq \ell\leq m$ and $1\leq u\leq \ell-1$ is injective, and its range is $[\frac{(m-1)m}{2}]$.
}
\begin{claim}\claimlab{claim:inj}
\felluinjective
\end{claim}

The proof of \claimref{claim:inj} is elementary by induction on $m$, and hence, we defer the proof to \appref{app:missingprooflowerbound}. What we are interested in is whether each $B_i$, for $i_2+1\leq i\leq t$, belongs to $\B_0$, $\B_1$, or $\B_2$. In \claimref{claim:blocks}, we prove that the blocks are mostly in $\B_2$ and the rest of the blocks are in $\B_1$, meaning that $\B_0=\emptyset$. In particular, we prove that for $2\leq\ell\leq m$ and $1\leq u\leq\ell-1$, $B_{i_2+\frac{(\ell-2)(\ell-1)}{2}+(\ell-u)}\in\B_1$ if and only if all of these conditions hold: (1) $\ell>2$, (2) $\ell$ is even, and (3) $u=\ell/2$. 

\begin{claim}\claimlab{claim:blocks}
For $2\leq\ell\leq m$ and $1\leq u\leq \ell-1$, 
\begin{itemize}[leftmargin=*]
    \item $B_{i_2+\frac{(\ell-2)(\ell-1)}{2}+(\ell-u)}\in\B_1$ if and only if $\ell>2, 2\mid\ell$, and $u=\ell/2$;
    \item Otherwise $B_{i_2+\frac{(\ell-2)(\ell-1)}{2}+(\ell-u)}\in\B_2$.
\end{itemize}
\end{claim}

We will give the proof of \claimref{claim:blocks} below and finish the proof of \lemref{lem:gs} first for readability. By \claimref{claim:blocks}, since there are only $\lfloor\frac{m}{2}\rfloor-1$ of such pairs of $(\ell,u)$ (satisfying $\ell>2, 2\mid\ell$, and $u=\ell/2$), we observe that $|\B_2|=\frac{(m-1)m}{2}-(\lfloor\frac{m}{2}\rfloor-1)$. Since we have that $B_i\in\B_1$ for all $i\in[i_2]$, we have $\B_0=\emptyset$ and therefore $|\B_0|=0$. This completes the proof of \lemref{lem:gs}.
\end{proof}

\begin{proof}[Proof of {\claimref{claim:blocks}}]
Recall that $S=S_{2,1}\circ4\circ S_{3,2}\circ4\circ S_{3,1}\circ4\circ\ldots\circ S_{m,m-1}\circ4\circ\ldots\circ S_{m,1}\circ4$ and $\str(B_{i_2+1})=S_{2,1}\circ4$, $\str(B_{i_2+2})=S_{3,2}\circ4$, $\str(B_{i_2+3})=S_{3,1}\circ4,\ldots,\str(B_t)=S_{m,1}\circ4$. 
For each $S_{\ell,u}$, we observe that $S_{\ell,u}$ is \emph{not} a substring of $S_{w'}$. Hence, we see that each block $B_i$ (for $i_2+1\leq i\leq t$) is roughly split into two blocks for the blocks of $\compress(w')$ unless it could copy beyond the character $4$. To observe the cases when this happens,
for each $S_{\ell,u}$, it is helpful to define:
\begin{itemize}
    \item $S_{\ell,u}^F\coloneqq\Enc(m-u+1)^2$ denotes the very first encoding concatenation that shows in $S_{\ell,u}$, 
    \item $S_{\ell,u}^{F,(1/2)}\coloneqq\Enc(m-u+1)$ denotes the very first encoding in $S_{\ell,u}$ (i.e., half of $S_{\ell,u}^F$),
    \item $S_{\ell,u}^L\coloneqq\Enc(m-u+\ell)^2$ denotes the very last encoding concatenation that shows in $S_{\ell,u}$, and
    \item For $k>u$, $S_{\ell,u}^{(k)}\coloneqq\Enc(m-u+1)^2\circ\Enc(m-u+2)^2\circ\ldots\circ\Enc(m)^2\circ2\circ\Enc(m+1)^2\circ\ldots\circ\Enc(m-u+k)^2$ denotes the first $k$ encoding concatenations that shows in $S_{\ell,u}$.
\end{itemize}
Then we observe the following claims. Since proofs of \claimref{claim:notrepeat} and \claimref{claim:repeat} are elementary, we defer the proofs to \appref{app:missingprooflowerbound}.

\begin{claim}\claimlab{claim:notrepeat}
$S_{\ell,u}^L\circ4\circ S_{\ell,u-1}^F$ does not repeat for different $\ell$ and $u$ such that $3\leq\ell\leq m$ and $2\leq u\leq\ell-1$.
\end{claim}

\claimref{claim:notrepeat} tells us that, due to the injectivity of the encoding, any block in $\compress(w')$ containing a portion of $S_{\ell,u}^L$ along with the delimiter `4' must finish at $S_{\ell,u}^{F,(1/2)}$ in the worst case. In particular, note that $S_{\ell,u}^L=\Enc(m-u+\ell)^2=S_{\ell+1,u+1}^L$ for $3\leq \ell<m$ and $2\leq u<\ell-1$. Moreover, we have $S_{\ell,u-1}^{F,(1/2)}=\Enc(m-u+2)$ and $S_{\ell+1,u}^{F,(1/2)}=\Enc(m-u+1)$, which can agree on all but the final bit (e.g., $S_{\ell,u-1}^{F,(1/2)}=00\cdots00$ and $S_{\ell+1,u}^{F,(1/2)}=00\cdots01$). Without the repetition of each encoding, a block might incorporate nearly the entire $S_{\ell+1,u}^{F,(1/2)}$ except for the last bit. Consequently, by having this last bit as a new character, $S_{\ell,u-1}\circ4$ would be placed in $\B_1$. Repeating the encoding twice eliminates this possibility and we can ensure that the scenario described in \claimref{claim:repeat} is the only case where type-1 blocks would occur. Again, see \appref{app:missingprooflowerbound} for the proof of \claimref{claim:repeat}.

\begin{claim}\claimlab{claim:repeat}
For $2\leq\ell\leq\lfloor\frac{m}{2}\rfloor-1$, $S_{\ell,1}^L\circ4\circ S_{\ell+1,\ell}$ repeats at $S_{2\ell,\ell+1}^L\circ4\circ S_{2\ell,\ell}^{(\ell+1)}$.
\end{claim}

Let's go back to the proof of \claimref{claim:blocks}. By \claimref{claim:repeat}, we can see that the block of the form $B_{i_2+\frac{(2\ell-2)(2\ell-1)}{2}+(2\ell-\ell)}$ which satisfies
\[\str\left(B_{i_2+\frac{(2\ell-2)(2\ell-1)}{2}+(2\ell-\ell)}\right)=S_{2\ell,\ell}\circ4,\]
is in $\B_1$, and all of the other blocks beyond $B_{i_2}$ are in $\B_2$. This completes the proof of \claimref{claim:blocks}.
\end{proof}

Taken altogether, we can lower bound the global sensitivity of the LZ77 compression scheme as stated in \thmref{thm:lowerbound} below.

\begin{theorem}\thmlab{thm:lowerbound}
Let $\compress:\Sigma^*\rightarrow\Sigma'^*$ be the LZ77 compression function. Then $\mathtt{GS}_\compress\geq 4^{-1/3}\cdot n^{2/3}\log^{1/3}n=\Omega(n^{2/3}\log^{1/3}n)$. 
\end{theorem}

\begin{proof}
Let $(w,w')\gets\QuinStr(m)$ and let $|w|=|w'|=n$. By \claimref{claim:length}, we have $|w|=|w'|=\Theta(m^3\log m)$ and therefore $n=\Theta(m^3\log m)$. Furthermore, \claimref{claim:length} tells us that there exists some $\alpha$ with $\frac{2}{3}\leq\alpha\leq 1$ such that $n=\alpha m^3\log m$.  
Now let $(B_1,\ldots,B_t)\gets\compress(w)$ and $(B'_1,\ldots,B'_{t'})\gets\compress(w')$.
Recall that if we look at the proof of \claimref{claim:set:blocksm2:GS}, it tells us that $t'-t=|\B_2|-|\B_0|$. From \lemref{lem:gs}, we have $|\B_0|=0$ and $|\B_2|=\frac{(m-1)m}{2}-(\lfloor\frac{m}{2}\rfloor-1)$, which implies that $t'-t=\frac{(m-1)m}{2}-(\lfloor\frac{m}{2}\rfloor-1)$. We know $|\compress(w)|=t(2\lceil\log n\rceil+\lceil\log|\Sigma|\rceil)$ and $|\compress(w')|=t'(2\lceil\log n\rceil+\lceil\log|\Sigma|\rceil)$, we have
\begin{align*}
   \mathtt{GS}_{\mathtt{Compress}} & \geq  \left||\compress(w)|-|\compress(w')|\right|\\  
   &= |t-t'|\left(2\lceil\log n\rceil+\lceil\log|\Sigma|\rceil\right) \\
    &= |t-t'|\left(2\left\lceil\log (\alpha m^3\log m)\right\rceil+\lceil\log|\Sigma|\rceil\right) \\
    &\geq \frac{m^2}{4}\cdot 4\log m = m^2\log m \ .
\end{align*}
Furthermore, since we have $n=\alpha m^3\log m$ for some $\frac{2}{3}\leq\alpha\leq 1$, we observe that
\begin{align*}
    m^2\log m &= m^2\cdot\frac{n}{\alpha m^3}= \frac{n}{\alpha}\cdot\frac{1}{m} = \frac{n}{\alpha}\cdot\frac{\alpha^{1/3}\cdot\log^{1/3}m}{n^{1/3}}\\
    &\geq \left(\frac{n}{\alpha}\right)^{2/3}\cdot 4^{-1/3}\cdot\log^{1/3}n \\
    &\geq 4^{-1/3}\cdot n^{2/3}\log^{1/3}n,
\end{align*}
where the first inequality comes from the observation $\log n = \log\alpha + 3\log m + \log\log m\leq 4\log m$ and the second inequality comes from $(1/\alpha)\geq 1$. Hence,
\begin{align*}
    \mathtt{GS}_{\mathtt{Compress}} &\geq m^2\log m \geq 4^{-1/3}\cdot n^{2/3}\log^{1/3}n,
\end{align*}
which completes the proof.
\end{proof}
\section{Is High Sensitivity Inevitable?}\seclab{sec:kolmogorov}

Based on the results of Lagarde and Perifel \cite{lagarde2018lempel}, one might wonder whether or not any effective compression mechanism will necessarily have \emph{high} global sensitivity. We argue that this is not necessarily the case by considering \emph{Kolmogorov compression}. The Kolmogorov compression of an input string $w$ is simply the encoding of the minimum-size Turing Machine $M$ such that the Turing Machine $M$ will eventually output $w$ when run with an initially empty input tape. It is also easy to see that the Kolmogorov compression scheme has low global sensitivity $\O{\log n}$  for strings of length $n$ as was previously noted in \cite{AFI23}. Let $\kc:\Sigma^*\rightarrow\mathfrak{M}$ be the Kolmogorov compression where $\mathfrak{M}$ is the set of all Turing machines. For a string $w\in\Sigma^n$, suppose that $\kc(w)=M\in\mathfrak{M}$. Then for a string $w'\sim w$ that differs on the $i\th$ bit, one can obtain $M'$ which outputs $w'$ by (1) running $M$ to obtain $w$, and (2) changing the $i\th$ character of $w$ to obtain $w'$. Then the description of $M'$ needs the description of $M$ and $i$ plus some constant. Since it takes $\log n$ bits to encode $i$, It follows that $\kc(w') \leq \left| M'\right| \leq |M| + \O{\log n} = \kc(w) + \O{\log n}$. In particular, the global sensitivity of Kolmogorov compression is upper bounded $\O{\log n}$. While  Kolmogorov compression is not computable, when it comes to efficient compression ratios, Kolmogorov compression is at least as effective as any other compression scheme, i.e., for any compression scheme $\compress$ there is a universal constant $C$ such that $\left|\kc(w)\right| \leq C + \left|\compress(w)\right|$ for all string $w \in \Sigma^*$.\footnote{The Turing Machine $M$ can implement any decompression algorithm. We can hardcode $z=\compress(w)$ to obtain a new machine $M_z$ which simulates $M$ on the input $z$ to recover $w$.} Thus, the goal of designing a compression algorithm with low global sensitivity does not need to be inconsistent with the goal of designing a compression algorithm that achieves good compression rates.

We note that one can construct a \emph{computable} variant of Kolmogorov compression $\ckc$ that preserves low global sensitivity and is competitive with any efficiently computable compression algorithm although the compression algorithm $\ckc$ itself is computationally inefficient. Instead of outputting the minimum-size Turing machine, one can output the \emph{minimum-score} Turing machine $M$ followed by $1\circ 0^{\log t_M}$ (where $t_M$ is the running time of the machine $M$), where the \emph{score} of a Turing machine $M$ is defined as $\score(M)\coloneqq|M|+1+\log t_M$.\footnote{Briefly, any string $w$ of length $n$ will have compression length at most $\O{n + \log n}$, i.e., consider the machine $M_w$ which hardcodes $w$ and copies it to the output tape in $\O{n}$ steps we have $\score(M_w) = \O{n + \log n}$. Thus, to find the minimum score TM, we only need to enumerate Turing machines with description length at most $\O{n + \log n}$ and we only need to run each of these TMs for at most $2^{\O{n + \log n}}$ steps. This will take exponential time, but the total amount of work is still finite and hence the compression algorithm is computable.} Since the length of the compression becomes the score of the minimum-score Turing machine, one can similarly argue that this computable variant of Kolmogorov compression has global sensitivity $\O{\log n}$. To see that the compression scheme is computable, note that any string $w$ of length $n$ will have compression length at most $\O{n + \log n}$, i.e., the machine $M_w$ which hardcodes $w$ and copies it to the output tape in $\O{n}$ has  $\score(M_w) = \O{n + \log n}$. Thus, to find the turing machine $M$ minimizing $\score(M)$ we only need to enumerate Turing machines with description length at most $\O{n + \log n}$ and we only need to run each of these TMs for at most $2^{\O{n + \log n}}$ steps since machines that run longer than this cannot have minimum score. While the compression take exponential time, it is computable as the total amount of work is finite.

It is also easy to see that the compression rate for our computable variant of Kolmogorov compression is at least as good as {\em any} efficiently computable compression algorithm, i.e., for any compression scheme $\compress$ running in time $\O{n^c}$ for some constant $c$ and {\em any} string $w \in \Sigma^n$ we have  $|\ckc(w)| \leq |\compress(w)| + \log n^c + \O{1} = |\compress(w)| + \O{\log n}$. This implies that a computable compression scheme which achieves good compression ratios does not necessarily have high global sensitivity. 

\newcommand{\kolcomp}{
Let $\ckc:\Sigma^*\rightarrow\mathfrak{M}\times\bin^*$ be the computable variant of Kolmogorov compression. Then $\GS_\ckc(n)=\O{\log n}$.
}
\begin{proposition}\proplab{prop:ckc}
    \kolcomp
\end{proposition}

\begin{proof}[Proof Sketch] (See \appref{app:kolcompression} for a more rigorous proof.)
    For a string $w\in\Sigma^n$, suppose that $\ckc(w)=(M,1\circ0^{t_M})$ where $M$ is the minimum-score Turing Machine that outputs $w$ running in time $t_M$. For $w'\sim w$ that differs on the $i\th$ bit, one can obtain a Turing Machine $M'$ that outputs $w'$ by running $M$ to obtain $w$ and then changing the $i\th$ character of $w$ to obtain $w'$. 
    Then the running time of $M'$ is $t_{M'}=t_M+\O{n}$. Furthermore, the description of $M'$ needs the description of $M$ plus $i$ plus some constant, which implies $|M'|\leq|M|+\O{\log n}$ since it takes $\log n$ bits to encode $i$.

    Let $\ckc(w')=(M'',1\circ0^{t_{M''}})$ where $M''$ is the minimum-score Turing Machine that outputs $w'$ with running time $t_{M''}$. Then by definition, we have $\score(M'')\leq\score(M')$. Hence,
    \begin{align*}
        |\ckc(w')| &= |M''|+1+\log t_{M''}   =\score(M'')\leq \score(M')  = |M'|+1+\log t_{M'} \\
         &\leq |M|+\O{\log n} + \log(t_M+\O{n})\leq |M|+ \log t_M + \O{\log n} \\
         &         =|\ckc(w)| + \O{\log n},
    \end{align*}
    which implies $\displaystyle\GS_\ckc(n)=\max_{w\in\Sigma^n}\max_{w':w\sim w'}||\ckc(w)|-|\ckc(w')||=\O{\log n}$. Note that $\log(t_M+\O{n})\leq\max\{\log t_M,\log\O{n}\}+1$ and the last inequality above holds whatever the maximum is between $\log t_M$ and $\log\O{n}$. 
\end{proof}

\FullVersion{
\def\shortbib{0}
\bibliographystyle{abbrv}
\bibliography{references,abbrev3,crypto}
\appendix
\allowdisplaybreaks
\section{Differentially Private Compression: Missing Proofs}\applab{app:dpcompress}

\begin{claim}
\claimlab{claim:preserves:DP}
Let $w\sim w'$ be any strings and $\epsilon,\delta>0$ be DP parameters. Let 
    \[ \BAD = \left\{\begin{array}{ll}
        \left[|\compress(w)|+1,|\compress(w')|+1\right] & \text{if } |\compress(w')| \geq |\compress(w)|;\\
        \left\{|\compress(w)|+1\right\} & \text{otherwise,}
    \end{array}\right.\]
    and $\lapalg(w,\epsilon,\delta)\coloneqq|\dpcompress(w,\epsilon,\delta)|$. Then $\Pr[\lapalg(w,\epsilon,\delta)\in \mathsf{BAD}] \leq \delta$.
\end{claim}

\begin{proof} 
Let $w \sim w'$ be given and let $\BAD' = [|\compress(w)|+1,|\compress(w)|+\GS_f + 1]$. Observe that we always have $\BAD \subseteq \BAD'$ since we must $|\compress(w)|+\GS_f +1 \geq \lcompress(w') + 1$ and  $|\compress(w)|+\GS_f + 1 \geq |\compress(w)|+1$.   
 We first observe that $\lapalg(w,\epsilon,\delta) \in \BAD'$ if and only if $Z \leq \GS_f+1-k$. The reason is as follows:
 \begin{itemize}
     \item If $\lapalg(w,\epsilon,\delta) \in \BAD'$, then $|\compress(w)|+1\leq |\compress(w)|+\max\{1,Z+k\}\leq |\compress(w)|+\GS_f+1$, which implies $Z+k\leq \GS_f+1$. Hence, $Z\leq \GS_f+1-k$.
     \item If $Z \leq \GS_f+1-k$, then $Z+k\leq \GS_f+1$, which implies $\max\{1,Z+k\}\leq \max\{1,\GS_f+1\} = \GS_f+1$. Hence, we observe $|\compress(w)|+1\leq |\compress(w)|+\max\{1,Z+k\}\leq |\compress(w)|+\GS_f+1$ and therefore $\lapalg(w,\epsilon,\delta) \in \BAD'$.
 \end{itemize}
Therefore, 
\begin{align}
\Pr[\lapalg(w,\epsilon,\delta) \in \BAD] &\leq  \Pr[\lapalg(w,\epsilon,\delta) \in \BAD'] \eqnlab{eq:1} \\
&=  \Pr[Z\leq \GS_f+1-k] \eqnlab{eq:2} \\
&= \Pr\left[Z\leq -\frac{\GS_f}{\epsilon}\ln\left(\frac{1}{2\delta}\right)\right] \eqnlab{eq:3} \\
&=\frac{1}{2}\exp\left(-\ln\left(\frac{1}{2\delta}\right)\right) \eqnlab{eq:4} \\
 &= \frac{1}{2}\exp(\ln(2\delta)) = \delta ,\nonumber 
\end{align}
where \eqnref{eq:1} follows because \(\BAD \subseteq \BAD'\), \eqnref{eq:2} follows by our observation that $\lapalg(w,\epsilon,\delta) \in \BAD'$ if and only if $Z \leq \GS_f+1-k$, \eqnref{eq:3} follows by definition of $k$, and \eqnref{eq:4} follows because $Z$ is sampled from the Laplace distribution $Z\sim \Lap(\frac{\GS_f}{\epsilon})$. 
\end{proof}

\begin{remindertheorem}{\thmref{theorem_dp}}
    \dptheorem
\end{remindertheorem}

\begin{proof}
    Let $w\sim w'$ be given and let ($\epsilon,\delta$) be the differential privacy parameters. Consider the set $\BAD$ as defined in \claimref{claim:preserves:DP}.

    We first claim that $\forall t \notin \BAD$, then $\Pr[\lapalg(w,\epsilon,\delta)=t]\leq e^{\epsilon}\Pr[\lapalg(w',\epsilon,\delta)=t]$. Note that if $t < |\compress(w)|+1$ then we have $\Pr[\lapalg(w,\epsilon,\delta)=t] = 0$ so the above inequality certainly holds. Otherwise, we have $ t > \max\{|\compress(w)|+1,\lcompress(w')+1\}$ since $t \not\in \BAD$ and we have 

        \begin{align*}
        \frac{\Pr[\lapalg(w,\epsilon,\delta)=t]}{\Pr[\lapalg(w',\epsilon,\delta)=t]} &=\frac{\Pr[|\compress(w)|+Z_1+k=t]}{\Pr[\lcompress(w')+Z_2+k=t]}\\
        &=
            \frac{\frac{\epsilon}{2\GS_f}\exp(\frac{-\epsilon}{\GS_f}|t-|\compress(w)|-k|)}{\frac{\epsilon}{2\GS_f}\exp(\frac{-\epsilon}{\GS_f}|t-\lcompress(w')-k|)}\\
            &=\exp\left[\frac{\epsilon}{\GS_f}\left(|\epsilon-\lcompress(w')-k|-|\epsilon-|\compress(w)|-k|\right)\right]  \\
    & \leq \exp\left[\frac{\epsilon}{\GS_f}|\lcompress(w')-|\compress(w)||\right] \\
    &\leq \exp(\epsilon) = e^{\epsilon} ,
        \end{align*}
        where $Z_1$ (resp. $Z_2$) is the random variable we sampled in $\dpcompress(w,\epsilon,\delta)$ (resp. $\dpcompress(w',\epsilon,\delta)$) and the last inequality follows because $||\compress(w)|-\lcompress(w')| \leq \GS_f$. Now let $S \subseteq \mathbb{N}$ be any subset of outcomes and note that 
        \begin{align}
            \Pr[\lapalg(w,\epsilon,\delta) \in S] &=  \Pr[\lapalg(w,\epsilon,\delta) \in S \backslash \BAD ] +  \Pr[\lapalg(w,\epsilon,\delta) \in S \cap \BAD ] \nonumber\\
            &\leq  \Pr[\lapalg(w,\epsilon,\delta) \in S \backslash \BAD ] +  \Pr[\lapalg(w,\epsilon,\delta) \in \BAD ] \nonumber\\
            & \leq \Pr[\lapalg(w,\epsilon,\delta) \in S \backslash \BAD ] + \delta \eqnlab{eq:pr:delta}\\
            & = \delta+ \sum_{t \in S \backslash \BAD} \Pr[\lapalg(w,\epsilon,\delta) =t] \nonumber \\
            & \leq \delta + \sum_{t \in S \backslash \BAD} e^{\epsilon} \Pr[\lapalg(w',\epsilon,\delta) =t] \eqnlab{eq:pr:delta_epsilon}\\
                        & \leq \delta +  e^{\epsilon} \Pr[\lapalg(w',\epsilon,\delta) \in S \setminus \BAD] \nonumber\\
                        & \leq e^{\epsilon} \Pr[\lapalg(w',\epsilon,\delta) \in S] + \delta ,\nonumber
        \end{align}
where \eqnref{eq:pr:delta} holds by \claimref{claim:preserves:DP}, moreover the condition for \eqnref{eq:pr:delta_epsilon} holds since the previous reasoning when $\Pr[\lapalg(w,\epsilon,\delta)=t]= e^{\epsilon}\Pr[\lapalg(w',\epsilon,\delta)=t]$, $ \forall t \notin \BAD $. Hence, $A_{\epsilon,\delta}$ is $(\epsilon,\delta)$-differentially private.
\end{proof}
\section{Upper Bound for the Global Sensitivity of the LZ77 Compression Scheme: Missing Proofs}\applab{app:missingproof}

\begin{definition}
\deflab{def:start_inside}
Let $\compress:\Sigma^*\rightarrow(\Sigma')^*$ be the LZ77 compression scheme and $w,w'$ be strings of length $n$. Let $(B_1,\ldots,B_t)\gets\compress(w)$ and $(B_1',\ldots,B_{t'}')\gets\compress(w')$. Then we say that for $k \in [t']$, the block $B'_{k}$ \emph{starts inside} $B_i$, $i \in[t]$, if and only if $s_{i} \leq s'_k \leq f_i$.

We can also define a predicate $\startinside:\mathbb{Z}\times\mathbb{Z}\rightarrow\bin$ for the blocks above, where
\[\startinside(i,k)\coloneqq\left\{\begin{array}{ll}
    1 & \text{if } i\in[t]\wedge k\in[t']\wedge s_i\leq s_k'\leq f_i,\text{ and}\\
    0 & \text{otherwise.}
    \end{array}\right.\]
    That is, $\startinside(i,k)=1$ if and only if $B_k'$ starts inside $B_i$.
    Let  $\mathcal{M}_i$ be the set of indices of blocks for compressing $w'$ which start inside $B_i$ defined as follows:     $\mathcal{M}_i\coloneqq \{k \in [t']: \startinside(i,k)=1 \}$.
We further define $\B_m$ to be the set of indices $i\in[t]$ (of blocks for compressing $w$) where the length of the set $\M_i$ equals $m$, i.e., $\B_m = \{i \in [t]: |\mathcal{M}_i|=m\}$.
\end{definition}

\subsection{Analyzing the Positions of Blocks}\applab{app:missingproofA}

\FullVersion{
\begin{remindertheorem}{\lemref{lemma:start_inside}}
    \lemstartinsidestatement
\end{remindertheorem}  

\begin{proof} Let us consider a particular block $B_i$ with start location $s_i$ and finish location $f_i$. Let $k\in[t']$ be the smallest index such that block $B_k'$ starts inside $B_i$ i.e., such that $s_i \leq s_k' \leq f_i$. Note that if no such $k$ exists then $B_i$ contains $0$ blocks and we are immediately done. 

Let $j$ be the index where $w$ and $w'$ differ, i.e., $w[j]\neq w'[j]$. Now we consider the following cases:

\begin{enumerate}

\label{case:a}    \item \emph{Case 0: \(f_i < j\)}. In this case, we argue the following claim.

    \begin{claim}\label{claim:c0}
        If $f_i<j$, then $B_k=B_k'\,\,\forall k\leq i$.
    \end{claim}
    \begin{proof}[Proof of Claim \ref{claim:c0}]
        As shown in Figure \ref{fig:case0}, we observe that $w[1,f_i]=w'[1,f_i']$ since $f_i<j$ and let $j$ be the index of the position where $w \sim w'$, particularly, $j$ is greater than $f_i$ and $f_j$. Hence, running the LZ77-block-based algorithm outputs the same result for both substrings, i.e., $B_k=B_k'$ for all $k\leq i$.
    \end{proof}

    \begin{figure}[ht!]  
    \centering
    \begin{tikzpicture}[scale=1.5]
        \draw (0,0) rectangle (6,0.5);
        \node at (-0.5, 0.25) {\(w\)};

        \draw[draw=red, thick] (0.9,0) rectangle (3.85,0.5);
        \draw[->, thick] (0.9,-0.5) -- (0.9,-0.2);
        \node[align=center, below] at (0.9,-0.6) {\(s_{i-1}\)};
        \fill[red!10] (0.9,0) rectangle (3.9,0.5);
        \node at (2.5,0.25) {$B_{i-1}$};

        \draw[draw=red, thick] (3.9,0) rectangle (5.25,0.5);
        \fill[red!10] (3.9,0) rectangle (5.25,0.5);
        
        \draw[->, thick] (3.9,-0.5) -- (3.9,-0.2);
        \draw[->, thick] (5.25,-0.5) -- (5.25,-0.2);
        \node at (3.9,-0.6) {$s_i$};
        \node at (5.25,-0.6) {$f_i$};

        \node at (4.6,0.25) {$B_i$};

        \draw[draw=black, thick, dashed] (5.8,-1.5) rectangle (5.8,0.5);
        \node[align=center, below] at (5.8,-2.1) {\( j\)};

        \begin{scope}[yshift=-1.5cm]
            \draw (0,0) rectangle (6,0.5);
            \node at (-0.5, 0.25) {\(w'\)};

            \draw[draw=red, thick] (0.9,0) rectangle (3.85,0.5);
            \draw[->, thick] (0.9,-0.5) -- (0.9,-0.2);
            \node[align=center, below] at (0.9,-0.6) {\(s'_{i-1}\)};
            \fill[red!10] (0.9,0) rectangle (3.9,0.5);
            \node at (2.5,0.25) {$B'_{i-1}$};

            \draw[draw=red, thick] (3.9,0) rectangle (5.25,0.5);
            \fill[red!10] (3.9,0) rectangle (5.25,0.5);
            
            \draw[->, thick] (3.9,-0.5) -- (3.9,-0.2);
            \draw[->, thick] (5.25,-0.5) -- (5.25,-0.2);
            \node at (3.9,-0.6) {$s'_i$};
            \node at (5.25,-0.6) {$f'_i$};

            \node at (4.6,0.25) {$B'_i$};
        \end{scope}
    \end{tikzpicture}
    \caption{Compression for $w$ and $w'$ which $f_i<j$.}
    \label{fig:case0}
\end{figure}

By Claim \ref{claim:c0}, we can conclude that there exists only one block $B_i'$ that can start inside $B_i$.

\label{case:b}    \item \emph{Case 1: \(s_i \leq j \leq f_i\)}. In this case, we observe the following claims:

    \begin{claim}\label{claim:c1si}
        If $s_i\leq j\leq f_i$, then $s_i=s_i'$.
    \end{claim}
    \begin{proof}[Proof of Claim \ref{claim:c1si}]
        We observe that $f_{i-1}<s_i\leq j$. Hence, by Claim \ref{claim:c0}, we have $B_k=B_k'$ for all $k\leq i-1$. This implies that $f_{i-1}=f'_{i-1}$. Hence, $s_i=f_{i-1}+1=f'_{i-1}+1=s_i'$.
    \end{proof}

    \begin{claim}\label{claim:c1fiprime}
        If $s_i\leq j\leq f_i$, then $f_i'\geq j$.
    \end{claim}
            \begin{figure}[ht!]
    \centering
    \begin{tikzpicture}[scale=1.5]
        \draw (0,0) rectangle (6,0.5);
        \node at (-0.5, 0.25) {\(w\)};

        \draw[draw=red, thick, fill=red!10] (0.6,0) rectangle (2.4,0.5);
        \draw[->, thick] (0.6,-0.5) -- (0.6,-0.1);
        \node[align=center, below] at (0.6,-0.6) {\(q_i\)};

        \draw[->, thick, red] (0.6,0.5) arc[start angle=180, end angle=0, radius=0.9cm];
        \path[->, thick, red] (0.6,0.5) edge[bend left] (1.6,0.5);
        
        \node[align=center, above] at (1.0,1.3) {\(\ell_i\)};
        \node[align=center, above] at (1.0,0.6) {\(\ell_{ij}\)};

        \draw[draw=red, thick, fill=red!10] (3.4,0) rectangle (5.2,0.5);
        
        \draw[->, thick] (5.3,1) -- (5.3,0.75);
        \node[align=center, below] at (5.3,1.6) {\(c_i\)};

        \node[align=center, above] at (4.0,1) {\(B_i=[q_i,\ell_i,c_i]\)};
        
        \draw[->, thick] (3.4,-0.4) -- (3.4,-0.2);
        \node[align=center, below] at (3.4,-0.4) {\(s_i\)};

        \draw[draw=gray, thick, fill=gray] (5.2,0) rectangle (5.4,0.5);

        \draw[->, thick] (5.4,-0.4) -- (5.4,-0.2);
        \node[align=center, below] at (5.4,-0.4) {\(f_i\)};

        \draw (0,-1.5) rectangle (6,-1);
        \node at (-0.5, -1.25) {\(w'\)};

        \draw[draw=red, thick, fill=red!10] (0.6,-1.5) rectangle (2.4,-1);

        \draw[draw=red, thick, fill=red!10] (3.4,-1.5) rectangle (5.2,-1);
        \draw (4.4,0) rectangle (4.6,0.5);

        \fill[gray!10] (4.4,0.0) rectangle (4.6,0.5);

        
        \draw[draw=gray,dashed] (3.4,-0.1) rectangle (5.4,0.6);
        

        \draw (3.4,-1.6) rectangle (4.6,-0.9);
        \draw (4.4,-1.5) rectangle (4.6,-1);
        \fill[black!20] (4.4,-1.5) rectangle (4.6,-1);

        \draw (4.6,-1.6) rectangle (5.9,-0.9);
        
        \draw[draw=gray,dashed] (3.4,-1.7) rectangle (5.4,-0.8);

        \node[align=center, below] at (4.0,-1.7) {\( B'_i\)};
        \node[align=center, below] at (5.0,-1.7) {\( B'_{i+1}\)};

        \draw[draw=gray, thick, fill=gray] (5.2,-1.5) rectangle (5.4,-1);
        
        \draw[draw=black, thick] (1.6,-2) -- (1.6,1);
        \fill[gray!10] (1.6,0.0) rectangle (1.8,0.5);
        \fill[gray!10] (1.6,-1.5) rectangle (1.8,-1);

        \draw[draw=black, thick] (4.4,-2) -- (4.4,1);

        \node[draw=black,align=center, below] at (4.4,-2.1) {\( j:w[j] \neq w'[j]\)};

    \end{tikzpicture}
    \caption{Compressing \(w\) and \(w'\) when the character is located at $j$-position such as $s_i \leq s'_{k} \leq f_i$ and $s_i \leq s_{k+1} \leq f_i$.}
    \label{fig:case1_c_1}
\end{figure}
    \begin{proof}[Proof of Claim \ref{claim:c1fiprime}]
        Suppose for contradiction that $f_i'<j$. We first observe that $w'[s_i',f_i'-1]$ is the longest substring of $w'[1,s_i'-1]$ that starts at $s_i'$. Next, we observe that $w[s_i,j-1]$ is a substring of $w[1,s_i-1]$. Note that, since $j$ is the only index where $w$ and $w'$ differ, we further observe that 
        \begin{equation}\label{eq1}
            w[s_i,j-1]=w'[s_i,j-1]=w'[s_i',j-1],
        \end{equation}
        where the last equality comes from Claim \ref{claim:c1si}. Similarly, since $s_i\leq j$, we have 
        \begin{equation}\label{eq2}
            w[1,s_i-1]=w'[1,s_i-1]=w'[1,s_i'-1].
        \end{equation}
        By Equation \eqref{eq1} and Equation \eqref{eq2}, along with the fact that $w[s_i,j-1]$ is a substring of $w[1,s_i-1]$ that we observed before, we have that $w'[s_i',j-1]$ is a substring of $w'[1,s_i'-1]$. Contradiction because $w'[s_i',j-1]$ is a \emph{longer} substring of $w'[1,s_i'-1]$ than $w'[s_i',f_i'-1]$ starting at $s_i'$! Hence, the LZ77 algorithm would have picked $w'[s_i',j-1]$ instead of $w'[s_i',f_i'-1]$ when constructing $B_i'$. This contradiction is due to the assumption that $f_i'<j$. Hence, we can conclude that $f_i'\geq j$.
        \end{proof}

\begin{claim}\label{claim:c1fprime+1geqfi}
    If $s_i\leq j\leq f_i$, then $f_{i+1}'\geq f_i$.
    \end{claim}

    \begin{figure}[ht!]  
    \centering
    \begin{tikzpicture}[scale=1.5]
        \draw (0,0) rectangle (6,0.5);
        \node at (-0.5, 0.25) {\(w\)};

        \draw[draw=red, thick] (0.9,0) rectangle (4.9,0.5);
        
        \draw[->, thick] (0.9,-0.5) -- (0.9,-0.2);
        \node[align=center, below] at (0.9,-0.6) {\(s_i\)};
        \draw[->, thick] (4.9,-0.5) -- (4.9,-0.2);
        \node[align=center, below] at (4.9,-0.6) {\(f_i\)};
        \fill[red!10] (0.9,0) rectangle (4.9,0.5);
        \node at (2.5,0.25) {$B_{i}$};

        \draw[draw=gray, thick, dashed] (3.0,-1.5) rectangle (3.3,0.5);
        \node[align=center, below] at (3.15,-1.5) {\( w[j]\)};

        \begin{scope}[yshift=-1.5cm]
            \draw (0,0) rectangle (6,0.5);
            \node at (-0.5, 0.25) {\(w'\)};

            \draw[draw=red, thick] (0.9,0) rectangle (3.5,0.5);
            
            \draw[->, thick] (0.9,-0.5) -- (0.9,-0.2);
            \node[align=center, below] at (0.9,-0.7) {\(s'_i\)};
            \fill[red!10] (0.9,0) rectangle (3.5,0.5);
            \node at (2.0,0.25) {$B'_{i}$};
            
            \draw[->, thick] (3.5,-0.5) -- (3.5,-0.2);
            
            \node at (3.5,-0.7) {$f'_i$};

            \fill[red!10] (3.5,0) rectangle (5.4,0.5);
            \node at (4.3,0.25) {$B'_{i+1}...$};
            
            \draw[draw=black, thick, dashed] (3.0,-0.0) rectangle (3.3,0.5);
        \end{scope}
    \end{tikzpicture}
    \caption{Compression for $w$ and $w'$ which $f'_{i+1}\geq j$.}
    \label{fig:case_1_c_2}
\end{figure}
    
\begin{proof}[Proof of Claim \ref{claim:c1fprime+1geqfi}]
    Suppose for contradiction that $f_{i+1}' < f_i$.  Observe that if $f_{i+1}' < f_i$ then we have $f_i' < s_{i+1}' < f_{i+1}' < f_i$. By definition of LZ77-Block-Based we observe that $w'[s_{i+1}',f_{i+1}'-1]$ was the longest possible substring of $w'[1,f_i']$ starting at $s_{i+1}'$. Next, we observe that $w[j+1,f_i-1]$ is a substring of $w[1, f_{i-1}]$. since $f_{i-1}<j$ we further observe that:
    \begin{equation}
    w[1, f_{i-1}]=w'[1, f_{i-1}]  ,    
    \label{eq:case1_b_eq1}
    \end{equation}
    and
    \begin{equation}
   w[j+1,f_i-1]=w'[j+1,f_i-1].
   \label{eq:case1_b_eq2}
    \end{equation}
By Equation \eqref{eq:case1_b_eq1} and Equation \eqref{eq:case1_b_eq2}, along with the fact that \( w[s_i, j-1] \) is a substring of \( w[1, s_i - 1] \) that we observed before, we note that \( w'[s_{i+1}', f_i - 1] \) is a substring of \( w'[1, f_{i-1}] \). However, since $f_{i-1} = f_{i-1}' < f_i'$ it follows that \( w'[s_{i+1}', f_i - 1] \) is also a substring of \( w'[1, f_i'] \). This is a contradiction as \( w'[s_{i+1}', f_i - 1] \) is longer than \( w'[s_{i+1}', f_{i+1}' - 1] \) which means that LZ77-Block-Based would have selected the longer block. Thus, \( f_{i+1}' \geq f_i \). 
\end{proof}
     
Taken together, by Claim \ref{claim:c1si}, Claim \ref{claim:c1fiprime}, and Claim \ref{claim:c1fprime+1geqfi}, we have $s_i'=s_i$ and $s_{i+2}'=f_{i+1}'+1\geq f_i +1 > f_i$. This implies that at most two blocks $B_i'$ and $B_{i+1}'$ can start inside $B_i$. In particular, for block $B_{i+2}'$ we have $s_{i+2}' > f_{i+1}' \geq f_i$. This completes the proof of Case 1.

\label{case:c} \item \emph{Case 2: $j<s_i$.} If $f'_{k}\geq f_i $ then block $B_{k+1}'$ does not start inside block $B_i$ as $s_{k+1}' > f_{k}' \geq f_i$ so that block $B_i$ trivially contains at most 1 block. Thus, we may assume without loss of generality that $f'_k<f_i$. 

Since the block $B_k'$ starts inside $B_i$ it is useful to define the offset $z \coloneqq s'_k-s_i$. In the figure \ref{fig:case_2_a} can be shown:

        \begin{figure}[ht!]  
    \centering
    \begin{tikzpicture}[scale=1.5]
        \draw (0,0) rectangle (7,0.5);
        \node at (-0.5, 0.25) {\(w\)};
        
        \draw[draw=red, thick] (0.3,0) rectangle (2.6,0.5);
        \draw[draw=red, thick] (3.5,0) rectangle (5.6,0.5);
        \fill[red!10] (3.5,0) rectangle (5.8,0.5);
        \draw[draw=red, thick] (5.8,0) rectangle (6.0,0.5);
        \fill[red!5] (5.8,0) rectangle (6.0,0.5);
        \node at (5.9,0.25) {$c_{i}$};

        \draw[->, thick] (3.5,1.2) -- (3.5,0.8);
        \node[align=center, below] at (3.5,1.6) {\(s_i\)};
        \draw[->, thick] (4.3,1.2) -- (4.3,0.8);
        \node[align=center, below] at (4.3,1.6) {\(s_{i}+z\)};

        \draw[->, thick] (6.0,1.2) -- (6.0,0.8);
        \node[align=center, below] at (6.0,1.6) {\(f_i\)};

        \node at (4.8,0.25) {$B_{i}$};

    \draw (0,-2) rectangle (7,-2.5);
    \fill[gray!10] (1.8,-2) rectangle (2.6,-2.5);
    \fill[gray!10] (1.8,0) rectangle (2.6,0.5);

    \node at (0.7,0.25) {$z$};
    \node at (2.2,0.25) {$\tilde{z}$};
    \node at (2.2,-2.25) {$\tilde{z}$};

    \node at (-0.5, -2.25) {\(w'\)};
    \draw[draw=red, thick] (1.0,-2) rectangle (1.8,-2.5);
    \draw[dashed] (1.0,0.5) rectangle (1.0,-2.0);
    \draw[dashed] (1.8,0.5) rectangle (1.8,-2.0);
    \draw[dashed,draw=gray] (1.3,0.8) rectangle (1.6,-2.8);
    \node[color=gray] at (1.5,-3.0) {$w'[j]$};

    \draw[draw=red, thick] (1.8,-2) rectangle (2.6,-2.5);
    \draw[draw=red, thick] (4.3,-2) rectangle (5.8,-2.5);
    \node[align=center, below] at (0.3,1.6) {\(q_i\)};
    \draw[->, thick] (0.3,1.2) -- (0.3,0.8);

    \node[align=center, below] at (1.0,1.6) {\(q_i+z\)};
    \draw[->, thick] (1.0,1.2) -- (1.0,0.8);

    \node[align=center, below] at (1.8,2.3) {\(q_i+z+\ell'_k\)};
    \draw[->, thick] (1.8,2.0) -- (1.8,0.8);
    
    \fill[red!10] (4.3,-2) rectangle (5.8,-2.5);
    
    \draw[draw=red, thick] (5.2,-2) rectangle (5.2,-2.5);
    \fill[red!5] (5.2,-2) rectangle (5.8,-2.5);

    \draw[->, thick] (4.3,-3.0) -- (4.3,-2.6);
    \node[align=center, below] at (4.3,-3.2) {\(s'_{k}\)};

    \draw[->, thick] (5.8,-3.0) -- (5.8,-2.6);
    \node[align=center, below] at (5.8,-3.2) {\(f'_{k}\)};
    
    \node at (4.8,-2.25) {$B'_{k}$};
    \node at (5.5,-2.25) {$B'_{k+1}$};

    \draw[dashed] (4.3,0.5) rectangle (4.3,-2.0);
    \draw[dashed] (5.2,0.5) rectangle (5.2,-2);
    \draw[dashed] (5.8,0.5) rectangle (5.8,-2.);

    \node[align=center, below] at (5.2,2.3) {\(s_{i}+z+\ell'_k\)};
    \draw[->, thick] (5.2,2.0) -- (5.2,0.8);

    \end{tikzpicture}
    \caption{Compression for $w$ and $w'$ which $j<s_i$}
    \label{fig:case_2_a}
\end{figure}


We observe the following claims:
    \begin{claim}
    \label{claim:case2:a}
    If $j<s_i$ then $j \leq q_i+z+\ell'_k$.
    \end{claim}

\begin{proof}[Proof of Claim \ref{claim:case2:a}]
    Suppose for contradiction that $j>q_i+z+\ell_k'$, then we note $w'[s'_k,f'_k-1] $ is the longest possible substring of $w'[1,s'_{k}-1]$ that starts at $s'_k$. Next we observe that $w[q_{i}+z,q_i+z+\ell_k']$ is trivially a substring of $w[1,j-1]$ since we assumed $j>q_i+z+\ell_k'$. Furthermore, we have $w[q_{i}+z,q_i+z+\ell_k']=w[s_{i}+z,s_i+z+\ell_k']$ since $s_i+z+\ell_k'=f_k'<f_i$. Therefore,        
        $w[s_{i}+z,s_i+z+\ell_k']$ is a substring of $w[1,j-1]$. We further observe that

        \begin{equation}
           w[s_i+z,s_i+z+\ell_k']=w[s_k',f_k']=w'[s_k',f_k']
           \label{equation:case2:a:1}
        \end{equation}
    
    and
    \begin{equation}
        w[1,j-1]=w'[1,j-1].
        \label{equation:case2:a:2}
    \end{equation}

By equation \eqref{equation:case2:a:1} and \eqref{equation:case2:a:2} we have that $w'[s'_k,f_k']$ is a substring of $w'[1,j-1]$, which is a substring of $w'[1,s_k'-1]$. 
Contradiction because $w'[s'_k,f_k']$ is a \emph{longer} substring of $w'[1,s'_{k'}-1]$ than $w'[s'_k,f'_k-1]$ starting at $s'_k$! Hence, the LZ77 algorithm would have picked $w'[s'_k,f_k']$ instead of $w'[s'_k,f'_k-1]$. This contradiction is due to the assumption that $j>q_i+z+\ell_k$. Hence, we can conclude that $j\leq q_i+z+\ell_k$.
\end{proof}

    \begin{claim}
    \label{claim:case2:b}
    If $j<s_i$ then $f'_{k+1}\geq f_i$.
    \end{claim}
    \begin{proof}[Proof of Claim 
    \ref{claim:case2:b}]

Suppose for contradiction that \( f'_{k + 1} < f_i \). 
By the definition of the LZ77-Block-Based compression algorithm, we first observe that $w'[s_{k+1}',f_{k+1}'-1]$ is the longest possible substring of $w'[1,s_{k+1}'-1]$ starting at $s_{k+1}'$. 
It is useful for our proof to define $\tilde{z} \coloneqq (f_i-1)-f_k'$. We observe that, by the definition of the LZ77-Block-Based compression algorithm, 
\begin{equation}\label{eqqisi}
    w[q_i + z + \ell'_k + 1, q_i + z + \ell'_k + \tilde{z}] = w[s_i + z + \ell'_k + 1, s_i + z + \ell'_k + \tilde{z}],    
\end{equation}
since the LHS of \eqref{eqqisi} is copied to the RHS of \eqref{eqqisi} when we run the algorithm. We also observe that
\begin{align}\label{eq:case2:11}
    s_i + z + \ell'_k + 1 &= s_i + (s_k'-s_i) + \ell'_k + 1 \nonumber\\
    &= s_k' + \ell_k' + 1 \nonumber \\
    &= f_k'+1 \nonumber \\
    &= s_{k+1}',
\end{align}
and
\begin{align}\label{eq:case2:22}
    s_i + z + \ell'_k + \tilde{z} &= s_i + (s_k'-s_i) + \ell'_k + (f_i-1)-f_k' \nonumber\\
    &= s_k' + \ell'_k + (f_i-1) - f_k'\nonumber \\
    &= f_k' + (f_i-1) - f_k' \nonumber\\
    &= f_i - 1.
\end{align}
Together with Equation \eqref{eq:case2:11} and \eqref{eq:case2:22}, by Equation \eqref{eqqisi} we have
\begin{equation}\label{c}
    w[q_i + z + \ell'_k + 1, q_i + z + \ell'_k + \tilde{z}] = w[s_{k+1}', f_i-1].
\end{equation}


Due to Claim \ref{claim:case2:a}, we have $j<q_i+z+\ell'_k+1$. Since $j$ was the unique index where $w[j]\neq w'[j]$, we have
\begin{equation}\label{a}
    w[q_i + z + \ell'_k + 1, q_i + z + \ell'_k + \tilde{z}] = w'[q_i + z + \ell'_k + 1, q_i + z + \ell'_k + \tilde{z}],
\end{equation}
and
\begin{equation}\label{b}
    w[s_{k+1}', f_i-1] = w'[s_{k+1}', f_i-1].
\end{equation}
Applying Equations \eqref{a} and \eqref{b} to Equation \eqref{c}, we obtain
\begin{equation}
    w'[q_i + z + \ell'_k + 1, q_i + z + \ell'_k + \tilde{z}] = w'[s_{k+1}', f_i-1].
\end{equation}
Since $w'[q_i + z + \ell'_k + 1, q_i + z + \ell'_k + \tilde{z}]$ is a substring of $w'[1,s_{k+1}'-1]$, we observe that $w'[s_{k+1}', f_i-1]$ is a substring of $w'[1,s_{k+1}'-1]$ starting at $s_{k+1}'$. Contradiction since $f_i-1 > f_{k+1}'-1$, which implies the LZ77-Block-Based compression algorithm would have picked $w'[s_{k+1}', f_i-1]$ instead of $w'[s_{k+1}', f_{k+1}'-1]$ when constructing the block $B_{k+1}'$. This contradiction is due to the assumption $f_{k+1}'<f_i$. Hence, we can conclude that $f_{k+1}'\geq f_i$.

Taken together, by Claim \ref{claim:case2:a} and Claim \ref{claim:case2:b}, we have $j\leq q_i+z+\ell'_k$ and $f'_{k+1}\geq f_i$. this implies that at most two blocks $B'_k$ and $B'_{k+1}$ can start inside $B_i$. In particular, for block $B'_{k+2}$ we have $s'_{k+2}\geq f'_{k+1}\geq f_i$. This completes the proof of Case 2.
\end{proof}


\end{enumerate}

We have considered \lemref{lemma:start_inside}, where in Case 0 we showed that if $f_i<j$, then $B_k=B_k'\,\,\forall k\leq i$ what means, in fact, for Claim \ref{claim:c0}  there is only a block $B'_i$ that start inside.
In Case 1 we showed in Claim \ref{claim:c1si} that If $s_i\leq j\leq f_i$, then $s_i=s_i'$, that explains when both start position are the same, so then there is only one block that can start inside $B_i$. However, we showed in Claim \ref{claim:c1fiprime} that if $s_i\leq j\leq f_i$, then $f_i'\geq j$ what means there are two blocks $B'_i$ and $B'_{i+1}$ that start inside $B_i$. Eventually, by Claim \ref{claim:c1fprime+1geqfi} if $s_i\leq j\leq f_i$, then $f_{i+1}'\geq f_i$ and this condition implies immediately since $s_i=s'_i$ that at most two blocks $B'_i$ and $B'_{i+1}$ can start inside $B_i$.
In Case 2, we proved in Claim \ref{claim:case2:a} that if $j<s_i$ then $j \leq q_i+z+\ell'_k$ and in Claim \ref{claim:case2:b} if $j<s_i$ then $f'_{k+1}\geq f_i$. We established that taking all possible assumptions of our lemma are valid. By proving that \emph{at most two blocks can start inside $B_i$} holds in all scenarios, we have proved that the lemma is true in general.
Therefore, we conclude that \lemref{lemma:start_inside} is proven.
\end{proof}
}{}

\newcommand{\claimGSstatement}{
Let $\compress:(\Sigma)^*\rightarrow(\Sigma')^*$ be the LZ77 compression function  and $w,w'$ be strings of length $n$ and $w\sim w'$. Then
    $\left| \left| \compress(w) \right| - \left| \compress(w') \right| \right| \leq  t_2(2\left\lceil \log n \right \rceil + \left\lceil \log |\Sigma|\right \rceil).$
}
\begin{claim}\claimlab{claim:compress:GS}
\claimGSstatement
\end{claim}

\begin{proof}[Proof of {\claimref{claim:compress:GS}}]
Let $(B_1,\ldots,B_t)\gets\compress(w)$ and $(B_1',\ldots,B_{t'}')\gets\compress(w')$. 
We first observe that there are $t$ blocks in $\compress(w)$ and encoding each block we need $2 \lceil\log n\rceil + \lceil\log|\Sigma|\rceil$ bits, since each block consists of two integers $q_i$ and $\ell_i$, and one character $c_i$ with $0\leq q_i,\ell_i\leq n-1$ and $c_i\in\Sigma$. Hence, $|\compress(w)| = t(2\lceil\log n\rceil + \lceil\log|\Sigma|\rceil).$ and 
$|\compress(w')| = t'(2\lceil\log n\rceil + \lceil\log|\Sigma|\rceil).$
Then,     $\left| \compress(w) \right| - \left| \compress(w') \right| =\left|t-t'\right||2\lceil\log n\rceil + \lceil\log|\Sigma|\rceil|.$ By \claimref{claim:set:blocksm2:GS} above, we observe that $|t-t'|\leq t_2$. Hence, we have \\$\left|\left| \compress(w) \right| - \left| \compress(w') \right|\right|\leq t_2(2\lceil\log n\rceil + \lceil\log|\Sigma|\rceil).$
\end{proof}

\subsection{Counting Type-2 Blocks for LZ77}\applab{app:missingproofC}

\begin{remindertheorem}{\lemref{lemma:length_cases_total}}
    \lemblocknumstatement
\end{remindertheorem}

\begin{proof}[Proof of {{\lemref{lemma:length_cases_total}}}]
To prove \lemref{lemma:length_cases_total}, we need the following claims.

\begin{claim}
\claimlab{claim:length_cases_total_a}
    If $i\in\B_2$ then either (1) $s_i\leq j\leq f_i$ or (2) $j-\ell_i < q_i \leq j$.
\end{claim}
\begin{proof}
    We first observe that it has to be either $j\geq s_i$ or $j<s_i$.
    \begin{itemize}
        \item If $j\geq s_i$, then we want to show that $s_i\leq j\leq f_i$. Suppose for contradiction that $j>f_i$. Then we have $w[1,f_i]=w'[1,f_i]$, i.e., two substrings are identical. Hence, by definition of the LZ77 compression algorithm, the block constructions are also the same, i.e., $s_i=s_i'$ and $f_i=f_i'$, where $s_i'$ and $f_i'$ denote the starting and finishing index for the block $B_i'$ for compressing $w'$. This implies that $\M_i=\{i\}$ and $i\in\B_1$. Contradiction since $i\in\B_2$. Hence, we have $s_i\leq j\leq f_i$ if $j\geq s_i$.
        \item If $j<s_i$, then we want to show that $j-\ell_i<q_i\leq j$, i.e., $q_i\leq j < q_i+\ell_i$.

        Suppose for contradiction that $j\not\in[q_i,q_i+\ell_i-1]$. Then we have 
        \begin{equation}\eqnlab{eq:C9C2}
            w[s_i,f_i-1]=w[q_i,q_i+\ell_i-1]=w'[q_i,q_i+\ell_i-1]=w'[s_i,f_i-1].
        \end{equation}
        Let $k$ be the smallest element in $\M_i$, i.e., $s_i\leq s_k'\leq f_i$ but $s_{k-1}'<s_i$. 
        \begin{itemize}
            \item If $s_k'=f_i$, then if there is another $k'\in\M_i$ with $k'\neq k$, then by choice of $k$, we have $k<k'$ and therefore $s_{k'}'>f_k'>s_k'=f_i$. Hence, $|\M_i|= 1$ and $i\in\B_1$. Contradiction!
            \item If $s_i\leq s_k'\leq f_i-1$, then from \eqnref{eq:C9C2} and from the fact that $f_i=s_i+\ell_i$, we observe that the substring $w'[s_i,f_i-1]$ can be obtained by shifting the substring $w'[q_i,q_i+\ell_i-1]$ by $s_i-q_i$. Hence,
            \begin{align*}
                w'[s_k',f_i-1] &= w'[s_k',s_i+\ell_i-1]\\
                &=w'[s_k'-(s_i-q_i),s_i+\ell_i-1-(s_i-q_i)]\\
                &=w'[q_i+(s_k'-s_i),q_i+\ell_i-1],
            \end{align*}
            which implies that $f_k'\geq f_i$ by definition of the LZ77 compression algorithm. If there is another $k'\in\M_i$ with $k'\neq k$, then by choice of $k$, we have $k<k'$ and therefore $s_{k'}'>f_k'\geq f_i$. Hence, $|\M_i|=1$ and $i\in\B_1$. Contradiction!
        \end{itemize}
        Hence, we have $j\in[q_i,q_i+\ell_i-1]$ and therefore $q_i\leq j<q_i+\ell_i$ if $j<s_i$.
    \end{itemize}
    Taken together, we can conclude that if $i\in\B_2$ then either $s_i\leq j\leq f_i$ or $q_i\leq j<q_i+\ell_i$ must hold.
\end{proof}

\begin{claim}
\claimlab{claim:length_cases_total_b}
    If $i_1,i_2\in\B_2$ then $(q_{i_1},\ell_{i_1})\neq(q_{i_2},\ell_{i_2})$. 
\end{claim}
\begin{proof}
Intuitively, we prove the claim by contradiction. Suppose there exists $i_1,i_2\in\B_2$ such that $(q_{i_1},\ell_{i_1})=(q_{i_2},\ell_{i_2})$. Then for the first block of $\compress(w')$ that starts inside $B_{i_2}$, we can copy a substring until at least the end of $B_{i_2}$ either from $w'[q_{i_2},q_{i_2}+\ell_{i_2}-1]$ or from $B_{i_1}$ since at least one of them does not contain the differing index $j$ by \claimref{claim:length_cases_total_a}. We remark that the same intuition applies for the LZ77 with self-referencing as well --- see \claimref{sr_uniqueoffset} in \appref{app:missingproofCSR} for details. Below we give the formal proof of \claimref{claim:length_cases_total_b}.

Suppose for contradiction that there exist distinct $i_1,i_2 \in \mathcal{B}_2 $ such that $(q_{i_1},\ell_{i_1})=(q_{i_2},\ell_{i_2})$. Without loss of generality, let $i_1<i_2$. Now, we consider the following cases:
\begin{itemize}
\label{case:length_cases_total:case:a} \item \emph{Case 1: $j<s_{i_1}$}. 
Since $i_1<i_2$, we observe that $j<s_{i_1}<s_{i_2}$. Let $k \in \M_{i_2}$ be the \emph{smallest} element in  $\M_{i_2}$. Since $q_{i_1}=q_{i_2}$ and $\ell_{i_1}=\ell_{i_2}$ and $j<s_{i_1}$, we have
$w[q_{i_1},q_{i_1}+\ell_{i_1}-1]=w[s_{i_1},f_{i_1}-1]=w'[s_{i_1},f_{i_1}-1]$ and $w[q_{i_2},q_{i_2}+\ell_{i_2}-1]=w[s_{i_2},f_{i_2}-1]=w'[s_{i_2},f_{i_2}-1]$. 
Therefore, we have
\begin{equation}\eqnlab{claim29eq1}
w'[s'_{k},f_{i_2}-1]=w'[s_{i_1}+(s'_{k}-s_{i_2}),f_{i_1}-1].    
\end{equation}

We claim $f'_k\geq f_{i_2}$. Suppose for contradiction that $f'_k<f_{i_2}$. By definition of LZ77, $w'[s'_k,f'_{k}-1]$ is the longest substring of $w'[1,s'_k-1]$ starting at $s'_k$. However, \eqnref{claim29eq1} implies that $w'[s'_{k},f_{i_2}-1]$ is also a substring of $w'[1,s'_k-1]$ which is \emph{longer than} $w'[s'_k,f'_{k}-1]$, which is a contradiction. Hence, we have $f'_k\geq f_{i_2}$.



If there is another $k' \in \M_{i_2}$ such that $k' \neq k$ then, by choice of $k$, we have $k<k'$ and therefore $s_{k'}' > f_k' \geq f_{i_2}$. This contradicts the definition of $\M_{i_2}$ since block $B'_{k'}$ does not start inside block $B_{i_2}$. It follows that $\left| \M_{i_2}\right| \leq 1$ and therefore $i_2\in\B_1$. This is a contradiction! 

\label{case:length_cases_total:case:b}  \item \emph{Case 2: $j \geq s_{i_1}$}.
We then observe $s_{i_1} \leq j<f_{i_1}$ (if $f_{i_1}\leq j$ then only one block would have started inside $B_{i_1}$ by observing $w[s_{i_1},f_{i_1}-1]=w'[s_{i_1},f_{i_1}-1]$, hence $i\in \B_1$, which is a contradiction since $i\in\B_2$). And since we assumed $i_1<i_2$, we further have $s_{i_1} \leq j<f_{i_1}<s_{i_2}$. 
Now since $s_{i_2}>j$, by \claimref{claim:length_cases_total_a}, we have $q_{i_2}\leq j<q_{i_2}+\ell_{i_2}$. Since $q_{i_1}=q_{i_2}$ and $\ell_{i_1}=\ell_{i_2}$ by our assumption, we then have $q_{i_1}\leq j<q_{i_1}+\ell_{i_1}$. This is a contradiction since $s_{i_1}\leq j<f_{i_1}$!

\end{itemize}
Taken together, we can conclude that if $i_1,i_2\in\B_2$ then $(q_{i_1},\ell_{i_1})\neq(q_{i_2},\ell_{i_2})$.   
\end{proof}

\begin{claim}
\claimlab{claim:length_cases_total_c}
    Let $\B_2^\ell\coloneqq\{i\in\B_2:\ell_i=\ell\}$ and suppose that $s_{i^*}\leq j\leq f_{i^*}$ for some $i^*\in[t]$. Then $|\B_2^\ell|\leq\ell$ for all $\ell\neq \ell_{i^*}$, and $|\B_2^{\ell_{i^*}}|\leq\ell_{i^*}+1$.
\end{claim}

\begin{proof}
We first observe from \claimref{claim:length_cases_total_a} that if $i\in\B_2$, either $s_i\leq j\leq f_i$ or $j-\ell_i<q_i\leq j$ holds. Let $i^*\in[t]$ be the unique index such that $s_{i^*}\leq j\leq f_{i^*}$.
\begin{itemize}
    \item If $\ell\neq\ell_{i^*}$, then for all $i\in\B_2^\ell$, we have $j-\ell_i<q_i\leq j$ since $j\not\in[s_i,f_i]$ for all $i\in\B_2^\ell$. Suppose for contradiction that there exists some $\ell'\neq\ell_{i^*}$ such that $|\B_2^{\ell'}|\geq \ell'+1$. Then we have at least $\ell'+1$ $i$'s such that $j-\ell'<q_i\leq j$. Then by the pigeonhole principle, there exists some $i_1,i_2\in\B_2^{\ell'}$ such that $q_{i_1}=q_{i_2}$, which implies that $(q_{i_1},\ell_{i_1}=\ell')=(q_{i_2},\ell_{i_2}=\ell')$. Contradiction by \claimref{claim:length_cases_total_b}! Hence, $|\B_2^\ell|\leq\ell$ for all $\ell\neq \ell_{i^*}$.
    \item If $\ell=\ell_{i^*}$, then we clearly see $i^*\in\B_2^{\ell_{i^*}}$. Then what is remained to show is that $|\B_2^{\ell_{i^*}}\setminus\{i^*\}|\leq\ell_{i^*}$. For all $i\in \B_2^{\ell_{i^*}}\setminus\{i^*\}$, we have $j-\ell_i<q_i\leq j$ since $s_{i^*}\leq j\leq f_{i^*}$. Now by the previous case analysis, we have $|\B_2^{\ell_{i^*}}\setminus\{i^*\}|\leq \ell_{i^*}$. Hence, $|\B_2^{\ell_{i^*}}|\leq\ell_{i^*}+1$.\qedhere
\end{itemize}
\end{proof}

Let's get back to the proof of  \lemref{lemma:length_cases_total}.
We first observe that block $B_i$ has length $\ell_i+1$ and the string length is $n$. Hence, we have $\sum_{i\in\B_2} (\ell_i+1)\leq n$. Furthermore, we observe that for $i\in\B_2^\ell$, the length of block $B_i$ is $\ell+1$. Therefore,
\[\sum_{i\in\B_2} (\ell_i+1) = \sum_\ell \sum_{i\in\B_2^\ell} (\ell+1) = \sum_\ell \left| \B_2^\ell \right| (\ell+1) \leq n.\]
Let $x_\ell \coloneqq \left| \B_2^\ell\right|$. Observe that $t_2 \leq \sum_\ell x_\ell$. Thus, to bound $t_2$ we want to maximize $\sum_\ell x_\ell$ subject to the constraints that (here, $i^*\in[t]$ is the \emph{unique} index where $s_{i^*}\leq j\leq f_{i^*}$)
\begin{enumerate}
    \item $x_\ell \leq \ell\text{ for all }\ell\neq\ell_{i^*}\text{ and }x_{\ell_{i^*}}\leq \ell_{i^*}+1$, and
    \item $\sum_\ell x_\ell  (\ell+1) \leq n$.\label{lengthcondition}
\end{enumerate}

It is clear that the maximum is obtained by setting $x_\ell$ the maximum possible value for all $\ell \leq z$ (set $x_\ell=\ell$ for all $\ell\leq z$ with $\ell\neq\ell_{i^*}$, and set $x_{\ell_{i^*}}=\ell_{i^*}+1$ if $\ell_{i^*}\leq z$) and $x_{\ell}=0$ for all $\ell > z$ for some threshold value $z$. It is followed by a simple swapping argument. Let $c_\ell=\ell$ if $\ell\neq\ell_{i^*}$ and $c_{\ell_{i^*}}=\ell_{i^*}+1$. Suppose we have a feasible solution $(x_0,\ldots,x_{n-1})$ which satisfies the constraints and we have $0<x_{\ell_1}<c_{\ell_1}$ and $0<x_{\ell_2}<c_{\ell_2}$ with $\ell_1<\ell_2$. Now define a new set of solution $(x_0',\ldots,x_{n-1}')$ where $x_{\ell_1}'=x_{\ell_1}+1, x_{\ell_2}'=x_{\ell_2}-1$, and $x_i'=x_i$ for all $i\neq \ell_1,\ell_2$. Then we observe that $0\leq x_\ell'\leq c_\ell$ for all $\ell=0,1,\ldots,n-1$, and furthermore,
\begin{align*}
    \sum_\ell x_\ell'(\ell+1) &= x_{\ell_1}'(\ell_1+1) + x_{\ell_2}'(\ell_2+1) + \sum_{\ell\neq\ell_1,\ell_2} x_\ell'(\ell+1)\\
    &= (x_{\ell_1}+1)(\ell_1+1) + (x_{\ell_2}-1)(\ell_2+1) + \sum_{\ell\neq\ell_1,\ell_2} x_\ell(\ell+1)\\
    &= \sum_\ell x_\ell(\ell+1) + (\ell_1-\ell_2) < \sum_\ell x_\ell(\ell+1) \leq n,
\end{align*}
which implies that $(x_0',\ldots,x_{n-1}')$ is also a feasible solution. We can repeat these swaps to keep the solution feasible but with a higher value in $x_\ell$ for a smaller index $\ell$, and finally, with having $x_\ell$ the maximum possible value for all $\ell\leq z$ for some threshold $z$. Note that $x_z$ could be less than $z$, i.e., for the threshold $z$, we have $x_z=c$ for some $0<c\leq z$.

To find the threshold $z$, we first observe that $\left[\sum_{\ell=0}^{z-1} \ell(\ell+1)\right] + x_z(z+1) = \sum_{\ell \geq 0} x_\ell (\ell+1) \leq n$ by condition (\ref{lengthcondition}). Simplifying the sum we have $\left[\sum_{\ell=0}^{z-1} \ell(\ell+1)\right] + x_z(z+1)=\frac{1}{3}(z-1)z(z+1) + x_z(z+1) \geq z^3/3$ where the last inequality follows because  $x_z\geq 1$. Thus, we have $z^3\leq 3n$ and $z\leq\sqrt[3]{3n}$. Setting this value of $z$, we have
\begin{equation*}
    t_2 \leq \sum_\ell x_\ell \leq \left(\sum_{\ell=0}^{\sqrt[3]{3n}} \ell\right) + 1 = \frac{\sqrt[3]{9}}{2} n^{2/3} + \frac{\sqrt[3]{3}}{2} n^{1/3} + 1,
\end{equation*}
where the addition of $1$ in the inequality above comes from the case if $\ell_{i^*}\leq z$.
\end{proof}

\subsection{Counting Type-2 Blocks for LZ77 with Self-Referencing}\applab{app:missingproofCSR}

For the LZ77 compression scheme with self-referencing, we can prove an analogous claim as stated in \lemref{LZ77srt2}.

\begin{lemma}\lemlab{LZ77srt2}
Let $\compress_\SR:\Sigma^*\rightarrow(\Sigma')^*$ be the LZ77 compression function with self-referencing and $w\sim w'$ be strings of length $n$. Let $(B_1,\ldots,B_t)\gets\compress_\SR(w)$ and $(B'_1,\ldots,B'_{t'})\gets\compress_\SR(w')$. Then $t_2\leq\frac{\sqrt[3]{9}}{2}n^{2/3}+\frac{\sqrt[3]{3}}{2}n^{1/3}+1$.
\end{lemma}

Following a similar proof strategy to \secref{typetwoblocks}, we verify that the same claims hold for LZ77 with self-referencing as well.

\begin{claim}\claimlab{sr_locationj}
For LZ77 with self-referencing, if $i\in\B_2$ then either (1) $s_i\leq j\leq f_i$ or (2) $j-\ell_i<q_i\leq j$.
\end{claim}

\begin{proof}
Since the only difference between the non-overlapping LZ77 and the LZ77 with self-referencing is that we allow for $\ell_i + q_i-1 > \ctc$ and this does not affect the proof at all, we can follow the exact same proof of \claimref{claim:length_cases_total_a}.
\end{proof}

\begin{claim}\claimlab{sr_uniqueoffset}
For LZ77 with self-referencing, if $i_1,i_2\in\B_2$ then $(q_{i_1},\ell_{i_1})\neq(q_{i_2},\ell_{i_2})$.
\end{claim}

\begin{proof}
We can follow a similar proof strategy with \claimref{claim:length_cases_total_b}, except that since we are considering LZ77 with self-referencing, in Case 1, $w'[s_k',f_k'-1]$ becomes the longest substring of $w'[1,f_k'-1]$. However, \eqnref{claim29eq1} implies that $w'[s_k',f_{i_2}-1]$ is also a substring of $w'[1,s_k'-1]$, which is a substring of $w'[1,f_k'-1]$. This is a contradiction since $w'[s_k',f_k'-1]$ does \emph{not} become the longest substring of $w'[1,f_k'-1]$. Hence, we have the same conclusion that $f_k'\geq f_{i_2}$, and the rest of the proof follows the same as \claimref{claim:length_cases_total_b}.
\end{proof}

\begin{claim}\claimlab{sr_b2lsize}
For LZ77 with self-referencing,
let $\B_2^\ell\coloneqq\{i\in\B_2:\ell_i=\ell\}$ and suppose that $s_{i^*}\leq j\leq f_{i^*}$ for some $i^*\in[t]$. Then $|\B_2^\ell|\leq\ell$ for all $\ell\neq \ell_{i^*}$, and $|\B_2^{\ell_{i^*}}|\leq\ell_{i^*}+1$.
\end{claim}

\begin{proof}
Since the proof of \claimref{claim:length_cases_total_c} relied on \claimref{claim:length_cases_total_a} and \claimref{claim:length_cases_total_b} along with the pigeonhole principle, we can follow the same proof strategy relying on \claimref{sr_locationj} and \claimref{sr_uniqueoffset} instead.
\end{proof}

\begin{proof}[Proof of {\lemref{LZ77srt2}}]
Since the proof of \lemref{lemma:length_cases_total} did not exploit the restriction that the non-overlapping LZ77 has, we can follow the same proof of \lemref{lemma:length_cases_total}.
\end{proof}

\begin{theorem}\thmlab{thm:thmupperboundsr}
    Let $\compress_\SR:(\Sigma)^*\rightarrow(\Sigma')^*$ be the LZ77 compression function with self-referencing, with unbounded sliding window size $W=n$. Then \[\GS_{\compress_\SR}\leq \left(\frac{\sqrt[3]{9}}{2} n^{2/3} + \frac{\sqrt[3]{3}}{2} n^{1/3} + 3\right) (2\lceil\log n\rceil+\lceil\log|\Sigma|\rceil) = \bigO{n^{2/3}\log n}.\]
\end{theorem}

\begin{proof}
As discussed in \secref{position}, we see that \lemref{lemma:start_inside} essentially holds similarly for the self-referencing case except that there might be one type-3 block (i.e., $t_3 \leq 1$), i.e., the block containing the edited position $j$, as stated in \cite[p.18]{AFI23}. Thus, we have $t=t_0+t_1+t_2+t_3$ and we further observe that $t'=\sum_{m=0}^3 m\cdot t_m=t_1+2t_2+3t_3$ which is similar to \claimref{claim:set:blocksm2:GS}. Hence, $t'-t=2t_3+t_2-t_0\leq t_2+2t_3\leq t_2+2$. Hence, by \lemref{LZ77srt2}, 
\begin{align*}
    \left|\left| \compress(w) \right| - \left| \compress(w') \right|\right| =& |t-t'|(2\lceil\log n\rceil+\lceil\log|\Sigma|\rceil) \nonumber \\
    \leq & \left(\frac{\sqrt[3]{9}}{2} n^{2/3} + \frac{\sqrt[3]{3}}{2} n^{1/3} + 3\right) (2\lceil\log n\rceil+\lceil\log|\Sigma|\rceil) .
\end{align*}
Since this inequality hold for arbitrary $w\sim w'$ of length $n$, we have
\begin{align*}
    \mathtt{GS}_{\mathtt{Compress}} &= \max_{w \in \Sigma^n}\max_{w' \in \Sigma^n~\mathtt{s.t.}~w\sim w'} \left| |\compress(w)| - |\compress(w')| \right|\\
    &\leq \left(\frac{\sqrt[3]{9}}{2} n^{2/3} + \frac{\sqrt[3]{3}}{2} n^{1/3} + 3\right) (2\lceil\log n\rceil+\lceil\log|\Sigma|\rceil) =\O{n^{2/3}\log n}.\qedhere
\end{align*}
\end{proof}

\subsection{Bounded Sliding Window}\applab{app:missingproofW}

\begin{reminderclaim}{\claimref{claim:newconstraint}}
    \newconstraint
\end{reminderclaim}

\begin{proof}
Let $B_i=[q_i,\ell_i,c_i]$ such that $s_i>j+W$. Then we will show that $B_i\in\B_0\cup\B_1$.
\begin{itemize}
    \item If $B_i\in\B_0$ then we clearly have $B_i\in\B_0\cup\B_1$.
    \item If $B_i\not\in\B_0$ then let $B'_{k}$ be the first block that starts inside $B_i$. We argue that $f'_{k}\geq f_i$, which implies that $B_i\in\B_1$. Suppose for contradiction that $f'_{k}<f_i$. Then we have that $w'[s'_{k},f'_{k}-1]$ is the longest substring of $w'[\max\{1,s'_{k}-W\},s'_{k}-1]$ that starts at $s'_{k}$. 
    Since $s_i>j+W$, we have that $j<s_i-W$ and therefore $w[q_i,q_i+\ell_i-1]=w'[q_i,q_i+\ell_i-1]$ since we only look at the previous $W$ characters when finding the longest substring to construct $B_i$. Hence,
    \begin{align*}
        w'[q_i+(s'_{k}-s_i),q_i+\ell_i-1] &= w[s_i+(s'_{k}-s_i),f_i-1]\\
        &=w[s'_{k},f_i-1]\\
        &=w'[s'_{k},f_i-1].
    \end{align*}
    Furthermore, since $q_i\geq s_i-W$ by definition of LZ77, we observe $q_i+(s'_{k}-s_i)\geq s'_{k}-W$, meaning that $w'[s'_{k},f_i-1]$ is also a substring of $w'[\max\{1,s'_{k}-W\},s'_{k}-1]$ that starts at $s'_{k}$. But since $f'_{k}<f_i$, $w'[s'_{k},f_i-1]$ is a \emph{longer} substring of $w'[\max\{1,s'_{k}-W\},s'_{k}-1]$ that starts at $s'_{k}$ than $w'[s'_{k},f'_{k}-1]$. Contradiction! This contradiction is due to the assumption that $f'_{k}<f_i$. Hence, we have $f'_{k}\geq f_i$ and therefore $B_i\in\B_1$. Therefore, $B_i\in\B_0\cup\B_1$.
\end{itemize}
Finally, if $B_i\in\B_2$ then $B_i\not\in\B_0\cup\B_1$, hence we have $s_i\leq j+W$.
\end{proof}

\begin{theorem}\thmlab{thm:thmupperboundsrW}
    Let $\compress_\SR:(\Sigma)^*\rightarrow(\Sigma')^*$ be the LZ77 compression function with self-referencing, with sliding window size $W$. Then \[\GS_{\compress_\SR}\leq \left(\frac{\sqrt[3]{81}}{2} W^{2/3} + \frac{\sqrt[3]{9}}{2} W^{1/3} + 5\right) (2\lceil\log n\rceil+\lceil\log|\Sigma|\rceil) = \bigO{W^{2/3}\log n}.\]
\end{theorem}

\section{Lower Bound for the Global Sensitivity of the LZ77 Compression Scheme: Missing Proofs}\applab{app:missingprooflowerbound}

\begin{reminderclaim}{\claimref{claim:length}}
    \claimlength
\end{reminderclaim}

\begin{proof}
We first observe that $|w|=|w'|$ since they only differ by one character ($2$ and $3$ in $S_w$ and $S_{w'}$, respectively). We then observe that $|S_w|=4m\lceil\log m\rceil + 2$ since we have $4m$ encodings with length $\lceil\log m\rceil$ each along with the character $2$ and $4$, which has length $1$ each. And similarly, for $2\leq\ell\leq m$ and $1\leq u\leq \ell-1$, we observe that $|S_{\ell,u}|=2\ell\lceil\log m\rceil + 1$. Hence,
\begin{align*}
    |w| &= |S_w| + \sum_{\ell=2}^m \sum_{u=1}^{\ell-1} \left(1+|S_{\ell,u}|\right) \\
    &= 4m\lceil\log m\rceil + 2 + \sum_{\ell=2}^m \sum_{u=1}^{\ell-1} \left(2 + 2\ell\lceil\log m\rceil\right)\\
    &= 4m\lceil\log m\rceil + 2 + \sum_{\ell=2}^m \big[2(\ell-1) + 2\ell(\ell-1)\lceil\log m\rceil\big] \\
    &= 4m\lceil\log m\rceil + 2 + (m-1)m + \lceil\log m\rceil \cdot\frac{2}{3}\left( m^3-m \right) \\
    &= \frac{2}{3}m^3\lceil\log m\rceil + \frac{10}{3} m\lceil\log m\rceil + (m-1)m+2.
\end{align*}
It is clear that $|w|>\frac{2}{3}m^3\lceil\log m\rceil$ since the term $\frac{10}{3} m\lceil\log m\rceil + (m-1)m+2$ is always positive for $m\geq 1$. Now we consider the upper bound. For $m=4$, we can directly show that $\frac{2}{3}m^3\lceil\log m\rceil + \frac{10}{3} m\lceil\log m\rceil + (m-1)m+2=\frac{378}{3}<128=m^3\lceil\log m\rceil$. For $m\geq 5$, we have $m^2\geq 25$ and $5(m-1)m+10<5(m-1)m+5m=5m^2\leq m^3$, which implies
\begin{align*}
    |w| &= \frac{2}{3}m^3\lceil\log m\rceil + \frac{10}{3} m\lceil\log m\rceil + (m-1)m+2\\
    &= \frac{2}{3}m^3\lceil\log m\rceil + \frac{10}{25}\cdot\frac{25}{3}m\lceil\log m\rceil + \frac{1}{5}\left[5(m-1)m+10\right]\\
    &< \frac{2}{3}m^3\lceil\log m\rceil + \frac{10}{25}\cdot\frac{1}{3}m^3\lceil\log m\rceil + \frac{1}{5}m^3\\
    &< \frac{2}{3}m^3\lceil\log m\rceil + \frac{10}{25}\cdot\frac{1}{3}m^3\lceil\log m\rceil + \frac{1}{5}m^3\lceil\log m\rceil = m^3\lceil\log m\rceil.
\end{align*}
Taken together, we can conclude that $|w|=|w'|=\Theta(m^3\log m)$.
\end{proof}

\begin{reminderclaim}{\claimref{claim:inj}}
For any integer $m\geq 2$, the function
\[f(\ell,u)\coloneqq\frac{(\ell-2)(\ell-1)}{2}+(\ell-u)\]
defined over integers $\ell$ and $u$ such that $2\leq \ell\leq m$ and $1\leq u\leq \ell-1$ is injective, and its range is $[\frac{(m-1)m}{2}]$.
\end{reminderclaim}

\begin{proof}
We prove this by induction on $m$. For the base case ($m=2$), the domain of the function $f$ only consists of a single element $\{(2,1)\}$, therefore the function is clearly injective and $f(2,1)=1$ shows that its range is $[1]=[\frac{(m-1)m}{2}]$. Now suppose that the claim holds for some $m=m'$ and consider the case where $m=m'+1$. By inductive hypothesis, we already know that $f(\ell,u)$ is an injective function for $2\leq\ell\leq m'$ and $1\leq u\leq \ell-1$ and its range is $[\frac{(m'-1)m'}{2}]$. When $\ell=m'+1$ and $1\leq u\leq \ell-1=m'$, we observe that
\begin{align*}
    f(\ell,u) &= f(m'+1,u) = \frac{(m'-1)m'}{2} + (m'+1-u),
\end{align*}
hence the output of the function is all different for different $u$'s, and its range is $[\frac{(m'-1)m'}{2}+1,\frac{(m'-1)m'}{2}+m']=[\frac{(m'-1)m'}{2}+1,\frac{m'(m'+1)}{2}]$. Since $[\frac{(m'-1)m'}{2}]$ and $[\frac{(m'-1)m'}{2}+1,\frac{m'(m'+1)}{2}]$ are mutually exclusive, we can conclude that the function is still injective over integers $\ell$ and $u$ such that $2\leq \ell\leq m'+1$ and $1\leq u\leq \ell-1$, and its range is $[\frac{(m'-1)m'}{2}]\cup[\frac{(m'-1)m'}{2}+1,\frac{m'(m'+1)}{2}] = [\frac{m'(m'+1)}{2}]$, which concludes the proof.
\end{proof}

\begin{reminderclaim}{\claimref{claim:notrepeat}}
$S_{\ell,u}^L\circ4\circ S_{\ell,u-1}^F$ does not repeat for different $\ell$ and $u$ such that $3\leq\ell\leq m$ and $2\leq u\leq\ell-1$.
\end{reminderclaim}

\begin{proof}
We have $S_{\ell,u}^L=\Enc(m-u+\ell)^2$ and $S_{\ell,u-1}^F=\Enc(m-u+2)^2$. Since the pairs $(-u+\ell,-u+2)$ for $3\leq\ell\leq m$ and $2\leq u\leq\ell-1$ are all distinct, the claim holds.
\end{proof}

\begin{reminderclaim}{\claimref{claim:repeat}}
For $2\leq\ell\leq\lfloor\frac{m}{2}\rfloor-1$, $S_{\ell,1}^L\circ4\circ S_{\ell+1,\ell}$ repeats at $S_{2\ell,\ell+1}^L\circ4\circ S_{2\ell,\ell}^{(\ell+1)}$.
\end{reminderclaim}

\begin{proof}
We observe that
\begin{align*}
    S_{2\ell,\ell+1}^L\circ4\circ S_{2\ell,\ell}^{(\ell+1)} &= \Enc(m-(\ell+1)+2\ell)^2\circ4\circ\Enc(m-\ell+1)^2\circ\ldots\circ\Enc(m)^2\circ2\circ\Enc(m+1)^2\\
    &=\Enc(m-1+\ell)^2\circ4\circ\Enc(m-\ell+1)^2\circ\ldots\circ\Enc(m)^2\circ2\circ\Enc(m+1)^2\\
    &=S_{\ell,1}^L\circ4\circ S_{\ell+1,\ell},
\end{align*}
which completes the proof.
\end{proof}
\allowdisplaybreaks
\section{Rigorous Proof for \propref{prop:ckc}}\applab{app:kolcompression}

Here we give a mathematically rigorous proof of \propref{prop:ckc}.

\begin{remindertheorem}{\propref{prop:ckc}}
\kolcomp
\end{remindertheorem}

\begin{proof}
    For a string $w\in\Sigma^n$, suppose that $\ckc(w)=(M,1\circ0^{t_M})$, i.e., $M$ is the minimum-score Turing Machine that outputs $w$ running in time $t_M$. For $w'\sim w$ that differs on the $i\th$ bit, one can obtain a Turing Machine $M'$ that outputs $w'$ as follows:
    \begin{itemize}
        \item Run $M$ to obtain $w$, and
        \item Flip the $i\th$ bit of $w$ to obtain $w'$.
    \end{itemize}
    Then the running time of $M'$ is $t_{M'}=t_M+\O{n}$, meaning that there exist a constant $C_1>0$ and a positive integer $n_1$ such that for all $n\geq n_1$ we have $t_{M'}\leq t_M+C_1n$. Furthermore, the description of $M'$ needs the description of $M$ plus $i$ plus some constant, which implies $|M'|\leq|M|+\O{\log n}$ since it takes $\log n$ bits to encode $i$. Also, this means that there exist a constant $C_2>0$ and a positive integer $n_2$ such that for all $n\geq n_2$ we have $|M'|\leq |M|+C_2\log n$.

    Let $\ckc(w')=(M'',1\circ0^{t_{M''}})$, i.e., $M''$ is the minimum-score Turing Machine that outputs $w'$ with running time $t_{M''}$. Then by definition, we have $\score(M'')\leq\score(M')$. Hence, for all $n\geq \max\{n_1,n_2\}$ we have
    \begin{align*}
        |\ckc(w')| &= |M''|+1+\log t_{M''}\\
        &=\score(M'')\\
        &\leq \score(M')\\
        &= |M'|+1+\log t_{M'}\\
        &\leq |M|+C_2\log n + 1 + \log(t_M+C_1 n).
    \end{align*}
    Now we consider the two cases:
    \begin{enumerate}
        \item If $t_M\geq C_1n$, then for all $n\geq\max\{n_1,n_2,2\}$ we have (recall that $\log$ is base 2)
        \begin{align*}
            |\ckc(w')| &\leq |M|+C_2\log n + 1 + \log(t_M+C_1 n)\\
            &\leq |M|+C_2\log n + 1 + \log(2t_M)\\
            &= |M|+1+\log t_M + 1 + C_2\log n\\
            &\leq |M|+1+\log t_M + (C_2+1)\log n\\
            &=|\ckc(w)|+(C_2+1)\log n.
        \end{align*}
        Hence, $\GS_\ckc(n)=\max_{w\in\Sigma^n}\max_{w':w\sim w'}||\ckc(w)|-|\ckc(w')||\leq (C_2+1)\log n$ for all $n\geq\max\{n_1,n_2,2\}$, which implies $\GS_\ckc(n)=\O{\log n}$.
        \item If $t_M<C_1n$, then for all $n\geq\max\{n_1,n_2,2C_1\}$ we have
        \begin{align*}
            |\ckc(w')| &\leq |M|+C_2\log n + 1 + \log(t_M+C_1 n)\\
            &\leq |M|+C_2\log n + 1 + \log(2C_1n)\\
            &=|M|+1+(C_2+2)\log n\\
            &\leq |M|+1+\log t_M+(C_2+2)\log n\\
            &=|\ckc(w)|+(C_2+2)\log n.
        \end{align*}
        Hence, $\GS_\ckc(n)=\max_{w\in\Sigma^n}\max_{w':w\sim w'}||\ckc(w)|-|\ckc(w')||\leq (C_2+2)\log n$ for all $n\geq\max\{n_1,n_2,2C_1\}$, which implies $\GS_\ckc(n)=\O{\log n}$.\qedhere
    \end{enumerate}
\end{proof}

}{
\chapter*{Appendices}
\appendix

\def\shortbib{0}
\bibliographystyle{abbrv}
\bibliography{references,abbrev3,crypto}

\begin{thebibliography}{10}

\bibitem{AFI23}
T.~Akagi, M.~Funakoshi, and S.~Inenaga.
\newblock Sensitivity of string compressors and repetitiveness measures.
\newblock {\em Inf. Comput.}, 291(C), Mar. 2023.

\bibitem{rfc7932}
J.~Alakuijala and Z.~Szabadka.
\newblock {Brotli Compressed Data Format}.
\newblock RFC 7932, July 2016.

\bibitem{CPM:BanChaRad24}
H.~Bannai, P.~Charalampopoulos, and J.~Radoszewski.
\newblock {Maintaining the Size of LZ77 on Semi-Dynamic Strings}.
\newblock In S.~Inenaga and S.~J. Puglisi, editors, {\em 35th Annual Symposium
  on Combinatorial Pattern Matching (CPM 2024)}, volume 296 of {\em Leibniz
  International Proceedings in Informatics (LIPIcs)}, pages 3:1--3:20,
  Dagstuhl, Germany, 2024. Schloss Dagstuhl -- Leibniz-Zentrum f{\"u}r
  Informatik.

\bibitem{rfc5282}
D.~L. Black and D.~McGrew.
\newblock {Using Authenticated Encryption Algorithms with the Encrypted Payload
  of the Internet Key Exchange version 2 (IKEv2) Protocol}.
\newblock RFC 5282, Aug. 2008.

\bibitem{blocki_et_al:LIPIcs.ICALP.2022.26}
J.~Blocki, E.~Grigorescu, and T.~Mukherjee.
\newblock {Privately Estimating Graph Parameters in Sublinear Time}.
\newblock In M.~Boja\'{n}czyk, E.~Merelli, and D.~P. Woodruff, editors, {\em
  49th International Colloquium on Automata, Languages, and Programming (ICALP
  2022)}, volume 229 of {\em Leibniz International Proceedings in Informatics
  (LIPIcs)}, pages 26:1--26:19, Dagstuhl, Germany, 2022. Schloss Dagstuhl --
  Leibniz-Zentrum f{\"u}r Informatik.

\bibitem{blocki_et_al:LIPIcs.APPROX/RANDOM.2023.59}
J.~Blocki, E.~Grigorescu, T.~Mukherjee, and S.~Zhou.
\newblock {How to Make Your Approximation Algorithm Private: A Black-Box
  Differentially-Private Transformation for Tunable Approximation Algorithms of
  Functions with Low Sensitivity}.
\newblock In N.~Megow and A.~Smith, editors, {\em Approximation, Randomization,
  and Combinatorial Optimization. Algorithms and Techniques (APPROX/RANDOM
  2023)}, volume 275 of {\em Leibniz International Proceedings in Informatics
  (LIPIcs)}, pages 59:1--59:24, Dagstuhl, Germany, 2023. Schloss Dagstuhl --
  Leibniz-Zentrum f{\"u}r Informatik.

\bibitem{cryptoeprint:2025/1733}
J.~Blocki, S.~Lee, and B.~S. Yepes-Garcia.
\newblock Differentially private compression and the sensitivity of {LZ77}.
\newblock Cryptology {ePrint} Archive, Paper 2025/1733, 2025.

\bibitem{CCS:Degabriele21}
J.~P. Degabriele.
\newblock Hiding the lengths of encrypted messages via gaussian padding.
\newblock In G.~Vigna and E.~Shi, editors, {\em ACM CCS 2021}, pages
  1549--1565. {ACM} Press, Nov. 2021.

\bibitem{rfc1951}
L.~P. Deutsch.
\newblock {DEFLATE Compressed Data Format Specification version 1.3}.
\newblock RFC 1951, May 1996.

\bibitem{TCC:DMNS06}
C.~Dwork, F.~McSherry, K.~Nissim, and A.~Smith.
\newblock Calibrating noise to sensitivity in private data analysis.
\newblock In S.~Halevi and T.~Rabin, editors, {\em TCC~2006}, volume 3876 of
  {\em {LNCS}}, pages 265--284. Springer, Berlin, Heidelberg, Mar. 2006.

\bibitem{farias2025differentiallyprivatemultiobjectiveselection}
V.~A.~E. Farias, F.~T. Brito, C.~Flynn, J.~C. Machado, and D.~Srivastava.
\newblock Differentially private multi-objective selection: Pareto and
  aggregation approaches, 2025.

\bibitem{Fis12}
D.~Fisher.
\newblock {CRIME Attack Uses Compression Ratio of TLS Requests as Side Channel
  to Hijack Secure Sessions}.
\newblock ThreatPost. Retrieved September 13, 2012.

\bibitem{GILRSU23}
S.~Giuliani, S.~Inenaga, Z.~Lipt\'{a}k, G.~Romana, M.~Sciortino, and C.~Urbina.
\newblock Bit catastrophes for the burrows-wheeler transform.
\newblock In {\em Developments in Language Theory: 27th International
  Conference, DLT 2023, Ume\r{a}, Sweden, June 12–16, 2023, Proceedings},
  page 86–99, Berlin, Heidelberg, 2023. Springer-Verlag.

\bibitem{GluHarPra13}
Y.~Gluck, N.~Harris, and A.~Prado.
\newblock {BREACH: Reviving the CRIME Attack}.
\newblock \url{https://breachattack.com/}, 2013.

\bibitem{Huffman52}
D.~A. Huffman.
\newblock A method for the construction of minimum-redundancy codes.
\newblock {\em Proceedings of the IRE}, 40(9):1098--1101, 1952.

\bibitem{lagarde2018lempel}
G.~Lagarde and S.~Perifel.
\newblock Lempel-ziv: a “one-bit catastrophe” but not a tragedy.
\newblock In {\em Proceedings of the Twenty-Ninth Annual ACM-SIAM Symposium on
  Discrete Algorithms}, pages 1478--1495. SIAM, 2018.

\bibitem{LZ77}
A.~Lempel and J.~Ziv.
\newblock A universal algorithm for sequential data compression.
\newblock {\em IEEE Transactions on Information Theory}, 23(3):337--343, 1977.

\bibitem{LZ78}
A.~Lempel and J.~Ziv.
\newblock Compression of individual sequences via variable-rate coding.
\newblock {\em IEEE Transactions on Information Theory}, 24(5):530--536, 1978.

\bibitem{10597947}
S.~Liu, Y.~Cao, T.~Murakami, J.~Liu, and M.~Yoshikawa.
\newblock { CARGO: Crypto-Assisted Differentially Private Triangle Counting
  Without Trusted Servers }.
\newblock In {\em 2024 IEEE 40th International Conference on Data Engineering
  (ICDE)}, pages 1671--1684, Los Alamitos, CA, USA, May 2024. IEEE Computer
  Society.

\bibitem{9762326}
Z.~Lu, H.~J. Asghar, M.~A. Kaafar, D.~Webb, and P.~Dickinson.
\newblock A differentially private framework for deep learning with convexified
  loss functions.
\newblock {\em IEEE Transactions on Information Forensics and Security},
  17:2151--2165, 2022.

\bibitem{palacios2022htb}
R.~Palacios, A.~F. Fern{\'a}ndez-Portillo, E.~F. S{\'a}nchez-{\'U}beda, and
  P.~Garc{\'\i}a-De-Z{\'u}{\~n}iga.
\newblock Htb: a very effective method to protect web servers against breach
  attack to https.
\newblock {\em IEEE Access}, 10:40381--40390, 2022.

\bibitem{paulsen2019debreach}
B.~Paulsen, C.~Sung, P.~A. Peterson, and C.~Wang.
\newblock Debreach: Mitigating compression side channels via static analysis
  and transformation.
\newblock In {\em 2019 34th IEEE/ACM International Conference on Automated
  Software Engineering (ASE)}, pages 899--911. IEEE, 2019.

\bibitem{ratliff2024framework}
Z.~Ratliff and S.~Vadhan.
\newblock A framework for differential privacy against timing attacks.
\newblock {\em arXiv preprint arXiv:2409.05623}, 2024.

\bibitem{RizDuo12}
J.~Rizzo and T.~Duong.
\newblock {The CRIME attack}.
\newblock Presentation at Ekoparty 2012, 2012.

\bibitem{sheng2025differentiallyprivatedistancequery}
W.~Sheng, J.~Chen, C.~Hu, B.~Cai, M.~Han, and J.~Yu.
\newblock Differentially private distance query with asymmetric noise, 2025.

\bibitem{song2024refined}
Y.~Song.
\newblock Refined techniques for compression side-channel attacks.
\newblock 2024.
\newblock \emph{(Master's Thesis, ETH Z{\"u}rich)}.

\bibitem{tetek:LIPIcs.APPROX/RANDOM.2024.73}
J.~T{\v{e}}tek.
\newblock {Additive Noise Mechanisms for Making Randomized Approximation
  Algorithms Differentially Private}.
\newblock In A.~Kumar and N.~Ron-Zewi, editors, {\em Approximation,
  Randomization, and Combinatorial Optimization. Algorithms and Techniques
  (APPROX/RANDOM 2024)}, volume 317 of {\em Leibniz International Proceedings
  in Informatics (LIPIcs)}, pages 73:1--73:20, Dagstuhl, Germany, 2024. Schloss
  Dagstuhl -- Leibniz-Zentrum f{\"u}r Informatik.

\bibitem{10735286}
A.~{\"{U}}nsal and M.~{\"{O}}nen.
\newblock Chernoff information as a privacy constraint for adversarial
  classification.
\newblock In {\em 2024 60th Annual Allerton Conference on Communication,
  Control, and Computing}, pages 1--8, 2024.

\bibitem{10.1145/3523227.3546781}
H.~Wang, D.~Zhao, and H.~Wang.
\newblock Dynamic global sensitivity for differentially private contextual
  bandits.
\newblock In {\em Proceedings of the 16th ACM Conference on Recommender
  Systems}, RecSys '22, page 179–187, New York, NY, USA, 2022. Association
  for Computing Machinery.

\bibitem{Welch84}
T.~Welch.
\newblock A technique for high-performance data compression.
\newblock {\em Computer}, 17(6):8--19, 1984.

\bibitem{wicker2024certificationdifferentiallyprivateprediction}
M.~Wicker, P.~Sosnin, I.~Shilov, A.~Janik, M.~N. Müller, Y.-A. de~Montjoye,
  A.~Weller, and C.~Tsay.
\newblock Certification for differentially private prediction in gradient-based
  training, 2024.

\bibitem{zhang2023sensitivityestimationdifferentiallyprivate}
M.~Zhang, X.~Liu, and L.~Yin.
\newblock Sensitivity estimation for differentially private query processing,
  2023.

\end{thebibliography}
}

\end{document}